\documentclass[journal,twoside,web]{ieeecolor}
\usepackage{generic}
\usepackage{defs}
\def\BibTeX{{\rm B\kern-.05em{\sc i\kern-.025em b}\kern-.08em
    T\kern-.1667em\lower.7ex\hbox{E}\kern-.125emX}}

\usepackage{algorithm}
\usepackage{algpseudocodex}

\usepackage{silence}
\newcommand{\jw}{j\omega}
\newcommand{\bbW}{\mathbb{W}}
\newcommand{\weights}{\begin{bsmallmatrix}I & \bbW\end{bsmallmatrix}}
\newcommand{\cA}{\mathcal{A}}
\newcommand{\cBM}[1][]{\mathcal{B}_{\scriptscriptstyle M\ifx#1\empty\else,#1\fi}}
\newcommand{\cBN}[1][]{\mathcal{B}_{\scriptscriptstyle N\ifx#1\empty\else,#1\fi}}
\newcommand{\hcA}{\widehat{\mathcal{A}}}
\newcommand{\hcBM}[1][]{\widehat{\mathcal{B}}_{\scriptscriptstyle M\ifx#1\empty\else,#1\fi}}
\newcommand{\hcBN}[1][]{\widehat{\mathcal{B}}_{\scriptscriptstyle N\ifx#1\empty\else,#1\fi}}
\newcommand{\cC}{\mathcal{C}}
\newcommand{\sC}[1][]{\mathscr{C}_{\ell\ifx#1\empty\else,#1\fi}}
\newcommand{\mycaption}[1]{\caption{\footnotesize #1 }}

\newcommand{\hstm}{\hspace{2em}}
\newcommand{\sm}{\text{-}}
\newcommand{\lcb}{\left\{}
\newcommand{\rcb}{\right\}}

\makeatletter
\newcommand{\safelabel}[1]{\ifcsname c@test#1\endcsname\nonumber\tag{\ref{#1}}\else\label{#1}\newcounter{test#1}\fi}
\makeatother

\newtheorem{problem}{Problem}
\newtheorem{notation}{Notation}
\newtheorem{assumption}{Assumption}
\DeclareMathOperator{\range}{range}

\begin{document}
\title{An iterative tangential interpolation algorithm for model reduction of MIMO systems}
\author{Jared Jonas and Bassam Bamieh, \IEEEmembership{Fellow, IEEE}
\thanks{Jared Jonas and Bassam Bamieh are  with the department of Mechanical Engineering at University of California, Santa Barbara, \texttt{\{jjonas,bamieh\}@ucsb.edu}. This work is partially supported by NSF award ECCS-2453491.}}

\maketitle

\begin{abstract}
    We consider model reduction of large-scale Multi-Input, Multi-Output (MIMO) systems using low-rank barycentric tangential interpolation on the frequency axis. Our scheme is related to the recently-developed Adaptive Antoulas--Anderson (AAA) algorithm, which is an iterative algorithm that uses concepts related to the Loewner framework. 
    One feature of our algorithm is  optimization of the interpolation weight matrices to minimize a proxy for the \(L_2\) system error. 
 %   Another feature 
%    Our algorithm has two main features. The first is the use of freedom in interpolation weight matrices to optimize a proxy for an \(L_2\) system error. The second 
%    
%    is the  use of low-rank interpolation, where we iteratively add low-order interpolation data based on several criteria including minimizing maximum errors.  
    We use the freedom in  interpolation point selection to provide multiple  minimum-error-based algorithms with trade-offs between computational complexity and approximation performance.  We prove that a weighted \(L_2\) norm of the proxy error system is non-increasing as interpolation points are added, and that the reduced-order system is stable whenever the original system is.  Computational results are provided demonstrating the closeness of the proxy system  and true system \(L_2\) errors. Comparisons with prior work demonstrate performance on par with standard model reduction methods.    
\end{abstract}


\section{Introduction}
Interpolatory model reduction methods have become an increasingly important tool in model-order reduction for many reasons, particularly their computational benefits~\cite{rafiq2022model}.  This is particularly important in the context of large-scale linear systems that may be very sparse, for example in systems derived from finite-difference or finite-element discretizations of phenomena represented by partial differential equations like fluid or structure mechanics~\cite{Lassila14,Hetmaniuk12}.   There are two important sub-classifications of interpolatory methods, which are moment matching methods and rational Krylov methods~\cite{rafiq2022model}.  Essentially, moment matching methods are a group of methods that create systems which match the given system's transfer function (or its derivatives) at a selected number of points in the complex plane.   

Moment matching for linear systems can be thought of as fitting a rational interpolant to frequency data, thus algorithms that construct rational interpolants from data can be used for model order reduction.  One such example of this is the Loewner framework originally introduced by Antoulas and Anderson~\cite{Antoulas86} and expanded upon thereafter~\cite{antoulas2017tutorial}.  From frequency response data, one can construct a state-space representation by forming certain matrices, namely the so-called Loewner matrix.  We will expand on this more when we talk about the Loewner framework for MIMO systems.  

Unlike algorithms where the interpolation points are pre-selected, the AAA (Adaptive Antoulas--Anderson) algorithm~\cite{Nakatsukasa_2018} introduced by Nakatsukasa, S\`ete, and Trefethen incorporates a recursive algorithm for selecting the next interpolation point at each step. It can be thought of as a greedy version of the scalar Loewner framework that iteratively builds up the interpolant point-by-point until the resulting rational function approximates the input data well.  At each step, the AAA algorithm selects the next interpolation point to be the point at which a certain error is maximized, then readjusts free parameters in the system to further minimize the error across the rest of the points.  Though introduced in the context of rational approximation, given a large set of points in the complex plane and the frequency response evaluated at those points, the AAA algorithm has proved itself able to generate satisfactory reduced order transfer functions for SISO systems.  

For both the Loewner framework and the AAA algorithm, there are generalizations that apply to MIMO systems.  For example, the Loewner framework has been generalized to MIMO systems with the use of tangential interpolation functions, which may be left- or right-tangential interpolants~\cite{antoulas2017tutorial,Antoulas2020}.  A function \(R\) is said to be a \textit{right-tangential interpolating function} if it satisfies the property \(R(s_i) V_i = G(s_i) V_i\) at each interpolation point \(s_i\), where \(V_i\) is a full column rank matrix.  Similarly, a function is said to be a \textit{left-tangential interpolating function} if it satisfies the property \(U_i^* R(s_i) = U_i^* G(s_i)\), where \(U_i\) is a full column rank matrix.  With the generalized Loewner framework, a reduced order system can be constructed by forming two Loewner matrices and performing a singular value decomposition (SVD) on each matrix, from which one can derive the coefficient matrices for a state space realization of a low order approximation.  This is computationally efficient and thus is very useful for large, sparse systems.  The iterative rational Krylov algorithm (IRKA)~\cite{Gugercin2008h2,Antoulas2020} instead fixes the number of interpolation points and relocates them until the reduced model satisfies the \(H_2\) optimality conditions, and Gugercin~\cite{Gugercin2008} combines a rational Krylov subspace with a Gramian-weighted one to obtain stable reduced models for single-input or single-output systems.  With the optimal weights derived in Section~\ref{weights.sec}, our reduced model coincides with the latter projection when \(p = 1\) (Remark~\ref{rem:projection}).  Both methods work with a fixed number of points, in contrast to the AAA-like methods we investigate here, which add a new interpolation point at each iteration.  

There have been several attempts at adapting the AAA algorithm for use in MIMO system reduction.  For example, the block-AAA algorithm adjusts the interpolant to use matrix valued functions and ``block'' interpolates (i.e. fully interpolates) the matrix-valued functions at each interpolation point~\cite{gosea2021algorithms,yu2023}.  This system's order grows by either the number of inputs or the number of outputs every iteration, as shown in our prior work~\cite{Jonas2024}.   Benner et al.~\cite{Benner23} introduced tangential-AAA, which uses tangential-interpolation from the Loewner framework to create a system that grows by a constant every iteration and thus performs better than the block-AAA algorithm~\cite{Benner23}.  

In this work, we aim to develop an algorithm based on our prior work and other AAA-like algorithms that is suitable for model reduction of MIMO linear, time-invariant (LTI) systems.   Available model reduction methods based on the AAA algorithm require a large set of evaluation points, and their performance is closely linked to both the number of points and the distribution of said points.  The resulting systems from these methods tend to be unstable, even when the system to be approximated is stable.  In  prior work~\cite{Jonas2024}, we began to address this with the system-AAA algorithm, which may be thought of as the limiting case of the block-AAA algorithm with an infinite number of evaluation points on the imaginary axis, eliminating the need for these evaluation points.  This algorithm gave satisfactory results when reducing single input, single output (SISO) systems, but performed relatively poorly when reducing MIMO systems, especially when the number of inputs and/or outputs is large.   This motivated the low-rank approximation algorithm introduced at the end of our prior work, which was a cursory look at some of the ideas we flesh out in this manuscript.

In this work we present a set of algorithms which generalizes and improves upon our prior work while laying out important theoretical implications.  The family of algorithms we introduce efficiently produces well-performing models in the MIMO case which are stable whenever the input system is, with performance which matches standard model reduction algorithms.  The contributions in the paper include the following:
\begin{itemize}
    \item A concise layout of left-tangential barycentric interpolants for MIMO systems, their state-space realizations (including a real-valued realization for conjugate-pair data), and sufficient conditions for regularity and interpolation.
    \item A closed-form choice of the free weight matrices as the minimizer of a weighted \(L_2\) error, computed from the state-space matrices and the matrix \(\Theta\) alone.  This is in contrast to Loewner- and AAA-based methods, which require a large number of frequency-response samples and an SVD of a Loewner matrix.  Our approach additionally corrects a mistake in our prior work~\cite{Jonas2024} that rendered the method ineffective for MIMO systems.
    \item Theorems showing that, with the aforementioned choice of weights, the reduced model is asymptotically stable whenever the input system is (Theorem~\ref{thm:R_nojwpoles}), a guarantee that the Loewner, block-AAA, and tangential-AAA methods do not offer.   Additionally, we show the weighted error is non-increasing as interpolation points are added and reaches zero once the state dimension reaches that of a minimal realization (Theorem~\ref{thm:norm_behavior}), with numerical evidence that it tracks the true error.  Finally, we connect the reduced model to Petrov--Galerkin projections and note that, for single-output systems, it coincides with the SVD--Krylov projection of~\cite{Gugercin2008} (Remark~\ref{rem:projection}).
    \item Three interpolation point selection strategies trading off computational cost and determinism, a rank refinement step for repeated frequencies, and an open-source implementation.
    \item Numerical comparisons with balanced truncation, the Loewner framework, and tangential-AAA, showing accuracy on par with balanced truncation and improved accuracy and stability relative to the other moment matching methods.
\end{itemize}

The layout of the paper is as follows.  We summarize the problem formulation and main results in the following Section~\ref{formulation.sec}.  Section~\ref{ss_real.sec} carefully catalogs the transfer functions and state space realizations for systems that appear in this paper, as there are important observability/controllability and MIMO pole/zero cancellation issues to keep track of.  Sections~\ref{weights.sec} and~\ref{inter_points.sec} form the bulk of our contribution. In Section~\ref{weights.sec} we show how the free parameters in the interpolant (the weight matrices) can be used to minimize a measure of \(L_2\) system error. The optimization problem has an explicit solution through a linear equation involving the weight matrices and certain system Gramians.  We also show that this system error is non-increasing across iterations.  Section~\ref{inter_points.sec} introduces three variations of the interpolation point selection algorithm, as well as the selection of approximation rank at said points.  Section~\ref{computations.sec} gives computational examples of the performance of our algorithms in terms of approximation error and resulting system stability, as well as a comparison to other model reduction algorithms, and finally provides some commentary on the results and algorithms as a whole.   

\section{Problem formulation and main results}							\label{formulation.sec}
We use the following standard notation for state-space realizations of a MIMO, linear time-invariant (LTI), noise-free system and its corresponding transfer function.
\begin{notation} \label{not:G}
    Let \(G\), a MIMO LTI system with no poles on the imaginary axis, have the state-space realization
        \begin{equation}
        G: \left\{ \begin{aligned}	 \dot{x} =& Ax + Bu \\ y =& Cx + Du		\end{aligned} \right.
        \hstm \Leftrightarrow \hstm  
        G(s) = \brac{\begin{array}{c|c}A & B \\ \hline C & D\end{array}}.			\label{G_defs.eq}
    \end{equation}
    This will be a $p\times q$ system with $n$ states.
\end{notation}

We use a left-tangential barycentric interpolant of the form 
\begin{multline}
            R(s) := M^{\sm1}(s)N(s) 
            :=\scalemath{0.7}{\p{I + \sum_{k=1}^\ell \frac{W_k U_k^*}{s-\jw_k}}^{\sm1}\hspace{-0.3em}\p{D + \sum_{k=1}^\ell \frac{W_k \Sigma_k V_k^*}{s-\jw_k}}} , 		\\
            := \p{\weights\,  \mathcal{M}(s)}\inv \weights \, \mathcal{N}(s), \label{eq:rtan_int}
\end{multline}
where \(W_k\in\comp^{p\times r_k}\) are parameters termed the ``weights''. The set $\lcb j\omega_1, \ldots, j\omega_\ell \rcb$ are the interpolation points (on the imaginary axis), and \(U_k\), \(\Sigma_k\), and \(V_k\) form a $r_k$-rank approximation of \(G(\jw_k) = U_k \Sigma_k V_k^*\) for some pre-specified $r_k$ obtained from the SVD of $G(\jw_k)$ as specified in Notation~\ref{amp:int_params}.  While restricting the interpolation points to the imaginary axis is not strictly necessary, we note that the system's dynamics are fully represented by its frequency response, thus it is sufficient to only consider strictly imaginary interpolation.  Furthermore, it reduces the frequency search space to only one dimension.  

Theorem~\ref{thm:tan_int}~\cite{antoulas2017tutorial,Benner23,Jonas2024}  below states that \(R(s)\) above is a \textit{left-interpolant} of \(G(s)\) at the points $\lcb \jw_k \rcb$, i.e. \(U_k^*G(\jw_k) = U_k^*R(\jw_k)\) for each \(\omega_k\) if \(R\) satisfies the regularity condition in Definition~\ref{def:reg_cond}.  We note that sufficient conditions guaranteeing regularity are detailed in the appendix in Lemma~\ref{lem:reg_cond}.

The second line in~\eqref{eq:rtan_int} states that the weights in \(M\) and \(N\) can be factored out into a 
matrix %\(\bbW\in\comp^{p\times r}\) 
\begin{equation}
    \bbW := \begin{bsmallmatrix}W_1 & \cdots & W_\ell\end{bsmallmatrix}\in\comp^{p\times r}\
\end{equation}
 containing all the weights, and  forming two block column systems \(\mathcal{M}\in\comp^{(p+r)\times p}\) and \(\mathcal{N}\in\comp^{(p+r)\times q}\), where 
 \(r = \sum_{k=1}^\ell r_k\).   This is detailed in Theorem~\ref{thm:R_ss} which also shows that \(R\) has a state-space representation with state dimension \(r\) that can be constructed from matrices consisting only of interpolation data and the matrix of free parameters.  In the case that \(G\) has real-valued matrices, it is desirable that \(R\) also has a real-valued realization.  We show in Theorem~\ref{thm:real_R} and Corollary~\ref{cor:W_conj} that a simple real-valued realization can be obtained when interpolating at conjugate pairs of frequencies with conjugate corresponding weights.

To use the interpolant~\eqref{eq:rtan_int} as a reduced order model, two questions must be answered.  First, how should one select the interpolation frequencies and their respective approximation ranks?  Additionally, what methods should be used to determine the values of the weight matrices?  The algorithm that we have developed is an iterative algorithm loosely based on the AAA algorithm that builds up a model by selecting one new interpolation frequency per iteration, then readjusts the weights by solving an optimization problem which minimizes the error across the frequency spectrum.  We will view these steps in more detail in their own sections, but we will begin to address them now, starting with the weight optimization.

Supposing a set of interpolation frequencies has already been selected.  We consider the \(L_2\) system norm
\[\norm{G}_2^2 := \frac{1}{2\pi}\int_{-\infty}^{\infty} \norm{G(\jw)}_F^2 \, \rmd\omega,\]
which is finite for any system with no poles on the imaginary axis and coincides with the \(H_2\) norm when the system is stable.  We optimize the free weight parameters in~\eqref{eq:rtan_int} to minimize the squared error across the frequency spectrum, weighted by \(M\):
\[\min_{\bbW} \norm{M(R-G)}_2^2.\]
Due to the manner in which the weights appear in~\eqref{eq:rtan_int}, this optimization problem can be stated as follows.  
\begin{problem} \label{pr:H2}
    \[\min_{\bbW} \norm{\begin{bsmallmatrix}I & \bbW\end{bsmallmatrix}(\smash{\underbrace{\mathcal{N}-\mathcal{M}G}_{H}})}_{2}^2.\]
\end{problem}

\bigskip\noindent
As shown in Theorem~\ref{thm:H_real}, the system \(H\) appearing in Problem~\ref{pr:H2} has the state space representation
\begin{equation}
    H := \mathcal{N} - \mathcal{M}G = \brac{\begin{array}{c|c}A & B \\ \hline -C & 0 \\ \phantom{-}\cC & 0\end{array}},
\end{equation}
where
\[\cC = \begin{bmatrix}U_1^* C(\jw_1 I - A)\inv \\ : \\ U_\ell^* C(\jw_\ell I - A)\inv\end{bmatrix} \in \comp^{r \times n},\]
thus we are able to recast Problem~\ref{pr:H2} as the following trace minimization problem: 
\begin{problem} \label{pr:tr}
    \begin{align*}
        &\min_{\bbW} \tr\p{\begin{bsmallmatrix}I & \bbW\end{bsmallmatrix} \begin{bsmallmatrix}C\Theta C^* & -C\Theta \cC^* \\ -\cC \Theta C^* & \cC \Theta \cC^*\end{bsmallmatrix} \begin{bsmallmatrix}I \\ \bbW^*\end{bsmallmatrix}}
    \end{align*}
\end{problem}

\bigskip\noindent
where \(\Theta := \frac{1}{2\pi}\int_{-\infty}^{\infty} (\jw I - A)\inv BB^*(\jw I - A)^{-*}\,\rmd\omega\), which coincides with the controllability Gramian of the original system when it is stable.  We show in Theorem~\ref{thm:What_sol} that Problem~\ref{pr:tr} and thus Problem~\ref{pr:H2} have explicit formulae for the global minimizer
\[\bbW^+ = C\Theta \cC^* \p{\cC \Theta \cC^*}\inv\]
whenever \(\cC\Theta\cC^*\) is invertible, a condition we state as Assumption~\ref{ass:CTC} in Section~\ref{weights.sec}.  Under it, the reduced-order system is stable whenever the input system is, so the regularity condition stated previously is always satisfied.  Furthermore, we show that with this choice on \(\bbW\), the error is non-increasing across subsequent iterations, i.e. \(\gamma^{(i+1)}(\bbW^{+(i+1)}) \leq \gamma^{(i)}(\bbW^{+(i)})\) where the superscript \((i)\) denotes the values at the \(i\)th iteration, and reaches zero once \(r = n_{\min}\), the size of a minimal realization of the input system.  Importantly, this implies that the \(L_2\) error norm between the input system and the reduced-order system also reaches zero once the reduced-order system's size reaches the size of a minimal realization of the input system. 

This weighted norm will not increase across iterations, thus there is freedom in how we choose the next interpolation frequency each iteration.  We introduce three methods, with the first being interpolating at the frequency of maximum error, i.e. the frequency at which the \(L_\infty\) norm of the error system is attained.  In this case, the level-set algorithm of Bruinsma and Steinbuch~\cite{Bruinsma90} is used at every step to find the frequency at which the \(H_\infty\) (or \(L_\infty\)) norm is attained.  This is computationally expensive, however, so we introduce two other methods, which we term the gridded and random approaches.  The three aforementioned methods form three algorithms that trade off computational complexity, approximation performance, and determinism. 

\section{State space realizations}									\label{ss_real.sec}
In our prior work, we began investigating what we termed low-rank interpolation~\cite{Jonas2024}, which uses a left-tangential interpolant.  In this work, we use this interpolant as a basis to derive a left-tangential interpolating system and a variant which has real-valued coefficients.   We note that similar descriptor realizations for tangential interpolants appear in the work of Antoulas and coauthors~\cite{Antoulas2020}, but a more concrete connection between these descriptor systems and barycentric interpolants first appeared in~\cite{Benner23}.  Following the derivation of the system representation, we show that an alternative real-valued realization exists when \(G\) is real-valued and conjugate interpolation data is used.  We then end the section by finding a state-space realization of \(H := \mathcal{N} - \mathcal{M}G\), which is used extensively in the optimization step.  

\subsection{Left-tangential interpolating systems}
We start by considering a left-tangential barycentric interpolant shown in equation~\eqref{eq:rtan_int}, which was introduced at the end of our previous work~\cite{Jonas2024}.  The definition for all of the matrices in the interpolant appear in Notation~\ref{amp:int_params} for convenience' sake.
\newcommand*{\ampintparamsstatement}{Consider the system \(G\) (as in Notation~\ref{not:G}).  Let \(\curly{\omega_i}_{i=1}^\ell\) be a set of distinct real numbers, and let \(\curly{r_i}_{i=1}^\ell\) be a corresponding set of ranks.  Consider the rank-\(r_i\)-truncated SVD of \(G(\cdot)\) at each \(j\omega_i\) where \(r_i\) is less than or equal to the rank of \(G(j\omega_i)\), i.e. \(U_i\Sigma_i V_i^*\) forms the best rank-\(r_i\) approximation of \(G(j\omega_i)\). Under these assumptions, we have full-column-rank \(U_i\in\comp^{p\times r_i}\), \(U_i^*U_i = I\), diagonal full-rank \(\Sigma_i\in\real^{r_i\times r_i}\),  and full-column-rank \(V_i\in\comp^{q\times r_i}\), and \(V_i^*V_i = I\).}
\begin{notation} \label{amp:int_params}
    \ampintparamsstatement
\end{notation}
We add a regularity condition detailed in Definition~\ref{def:reg_cond}, which ensures the interpolant has finite values at the interpolation points.
\begin{definition} \label{def:reg_cond}
    We consider the system \(R\) in eq.~\eqref{eq:rtan_int} to be \textit{regular} if its transfer function is finite at all interpolation points \(\curly{\jw_1, \ldots, \jw_\ell}\). 
\end{definition}
Sufficient conditions for this regularity condition are laid out in Lemma~\ref{lem:reg_cond} in the Appendix; as noted in Remark~\ref{rem:Wi_full}, these conditions also guarantee that each \(W_i\) has full column rank, as required by Theorem~\ref{thm:tan_int}.  As shown in Theorem~\ref{thm:tan_int}, this interpolant, assuming the regularity condition, satisfies the left-interpolation property of \(U_k^* G(\jw_k) = U_k^* R(\jw_k)\) and that \(R(s)\) tends towards \(D\) as \(s\) approaches infinity.
\newcommand*{\thmtanintstatement}{Consider the matrices defined in Notation~\ref{amp:int_params} and the interpolant \(R\) from equation~\eqref{eq:rtan_int},  where each \(W_i\in\comp^{p\times r_i}\), \(i=1,\ldots, \ell\) is a full-column-rank free parameter which satisfies the regularity condition of \(R(j\omega_i)\) being finite as defined in Definition~\ref{def:reg_cond}.  For each interpolation frequency \(\omega_i\), \(R\) left-interpolates \(G\).  Furthermore, \(\lim_{|s|\to\infty} R(s) = D\).}
\begin{theorem} \label{thm:tan_int}
    \thmtanintstatement  
\end{theorem}
 
\begin{proof}
    Proof in the Appendix.
\end{proof}

Now, we will provide state-space realizations for the objects appearing equation~\eqref{eq:rtan_int}.  The system \(R(s)\) is the product of two systems \(M\inv\), and \(N\).  Both \(M\) and \(N\) contain terms which are left-multiplied by a weight, thus they may each be written as a product between the weight matrix \(\bbW\) and a block-column system, i.e. \(M = \weights\mathcal{M}\) and \(N = \weights\mathcal{N}\), where \(\bbW = \begin{bsmallmatrix}W_1 & \cdots & W_\ell\end{bsmallmatrix}\), 
\begin{align*}
    \mathcal{M}(s) &:= \begin{bsmallmatrix}I \\ \mathcal{M}_1(s) \\ : \\ \mathcal{M}_\ell(s)\end{bsmallmatrix}, \; \text{and} \; \mathcal{N}(s) := \begin{bsmallmatrix}D \\ \mathcal{N}_1(s) \\ : \\ \mathcal{N}_\ell(s)\end{bsmallmatrix}.
\end{align*}
Each sub-system in \(\mathcal{M}\) and \(\mathcal{N}\), 
\[\mathcal{M}_k := \frac{U_k^*}{s-\jw_k}, \; \text{and }\; \mathcal{N}_k := \frac{\Sigma_k V_k^*}{s-\jw_k},\]
have readily apparent state-space realizations, namely
\begin{equation} \label{eq:comp_mnk}
    \mathcal{M}_k = \brac{\begin{array}{c|c}j\omega_k I & U_k^* \\ \hline I & 0\end{array}}, \; \mathcal{N}_k = \brac{\begin{array}{c|c}j\omega_k I & \Sigma_k V_k^* \\ \hline I & 0\end{array}}.
\end{equation}

State-space realizations for \(\mathcal{M}\), \(\mathcal{N}\), \(M\), and \(N\) can now be constructed by concatenating each of these systems.  Indeed, 
\begin{align}
    \mathcal{M} &= \brac{\begin{array}{c|c}\mathcal{A} & \cBM \\ \hline 0 & I \\ I & 0\end{array}}, \; & \mathcal{N} &= \brac{\begin{array}{c|c}\mathcal{A} & \cBN \\ \hline 0 & D \\ I & 0\end{array}}, \label{eq:MNcal_ss} \\
    M &= \brac{\begin{array}{c|c}\mathcal{A} & \cBM \\ \hline \bbW & I\end{array}}, \; & N &= \brac{\begin{array}{c|c}\mathcal{A} & \cBN \\ \hline \bbW & D\end{array}} \label{eq:MN_ss}.
\end{align}
where
\[\cA := \mathrm{diag}(\jw_1 I, \ldots, \jw_\ell I), \; \begin{bsmallmatrix}\cBM & \cBN\end{bsmallmatrix} := \begin{bsmallmatrix}U_1^* & \Sigma_1 V_1^* \\ : & : \\ U_\ell^* & \Sigma_\ell V_\ell^*\end{bsmallmatrix}\]
are matrices that consist only of interpolation data.  From these realizations, it is possible to find a realization of \(R\) as shown in Theorem~\ref{thm:R_ss}.  The realizations of \(M\) and \(N\) both have poles at the interpolation points, but these cancel in the product \(M\inv N\): as the proof of Theorem~\ref{thm:R_ss} shows, the corresponding modes are unobservable, and the poles of \(R\) are instead the eigenvalues of \(\cA - \cBM\bbW\), which depend on the weights and are studied in Theorem~\ref{thm:R_nojwpoles}.  Importantly, this shows it is possible to form a state-space realization of the system by forming three matrices from interpolation data (\(\cA\), \(\cBM\), and \(\cBN\)) and one matrix \(\bbW\) from the free parameters.  

\begin{theorem} \label{thm:R_ss}
    Let \(M\) and \(N\) be defined as in equation~\eqref{eq:MN_ss} and let \(R = M\inv N\).  Then, 
    \begin{equation} \label{eq:R_ss}
        R = \brac{\begin{array}{c|c}\cA-\cBM\bbW & \cBN - \cBM D \\ \hline \bbW & D \end{array}}.
    \end{equation}
\end{theorem}
\begin{proof}
    Because \(M\) has an invertible feedthrough matrix, we can find a state space representation for \(M\inv\), which is
    \[M\inv = \brac{\begin{array}{c|c}\cA - \cBM  \bbW & \cBM  \\ \hline - \bbW & I \end{array}}.\]
    Thus we can find a combined state space realization for \(R\) by cascading the realizations for \(M\inv\) and \(N\), giving
    \[R = \brac{\begin{array}{cc|c}\cA & 0 & \cBN \\ \cBM \bbW & \cA - \cBM \bbW & \cBM D \\ \hline \bbW & -\bbW & D\end{array}}.\]
    Performing the coordinate change \(z = Tx\) with \(T = \begin{bsmallmatrix}I & -I \\ I & I\end{bsmallmatrix}\), whose inverse is \(\frac{1}{2}\begin{bsmallmatrix}I & I \\ -I & I\end{bsmallmatrix}\), yields
    \[R = \brac{\begin{array}{cc|c}\cA-\cBM\bbW & 0 & \cBN - \cBM D \\ \cBM \bbW & \cA & \cBN + \cBM D \\ \hline \bbW & 0 & D\end{array}}.\]
    From this it is clear that the second block mode is unobservable and can be eliminated, yielding equation~\eqref{eq:R_ss}.
\end{proof}

\subsection{Enforcing real coefficients}
In the case that \(G\) has real coefficients, we see the realization of \(R\) shown in equation~\eqref{eq:R_ss} will not have real coefficients in general due to the presence of complex values in \(\cA\), \(\cBM\), and \(\cBN\).  In order to obtain a real-valued realization of the data from the interpolant, we must interpolate both positive and negative frequencies (i.e. \(\pm j\omega_k\) for positive \(\omega_k\)) while ensuring the weights associated with this conjugate pair are themselves conjugate, as shown in Theorem~\ref{thm:real_R} and Corollary~\ref{cor:W_conj}.

\begin{theorem} \label{thm:real_R}
    Consider the data defined in Notation~\ref{amp:int_params} with the restriction that \(\omega_k \geq 0\) and let \(G\) have real coefficients.  Let the reduced-order system \(R\) be constructed with the interpolation data matrices
    \[\cA = \mathrm{diag}\{\cA_1, \ldots, \cA_m\}\;\text{and}\;\begin{bsmallmatrix}\cBM & \cBN\end{bsmallmatrix} = \begin{bsmallmatrix}\cBM[1] & \cBN[1] \\ : & : \\ \cBM[m] & \cBN[m]\end{bsmallmatrix},\]
    where \(m\) is the number of non-negative frequencies, and
    \[\cA_k = 0, \, \cBM[k] = U_k\trans, \; \text{and}\; \cBN[k] = \Sigma_k V_k\trans\]
    when \(\omega_k = 0\), and
    \[\cA_k = \begin{bsmallmatrix}\jw_k I & 0 \\ 0 & -\jw_k I\end{bsmallmatrix} \; \text{and} \; \begin{bsmallmatrix}\cBM[k] & \cBN[k]\end{bsmallmatrix} = \begin{bsmallmatrix}U_k^* & \Sigma_k V_k^* \\ U_k\trans & \Sigma_k V_k\trans\end{bsmallmatrix}\]
    otherwise, ensuring frequencies are interpolated in conjugate pairs.   
    
    \bigskip
    Define the \textit{real-valued} interpolation data matrices
    \[\hcA := \mathrm{diag}\{\hcA_1, \ldots, \hcA_m\}\;\text{and}\;\begin{bsmallmatrix}\hcBM & \hcBN\end{bsmallmatrix} := \begin{bsmallmatrix}\hcBM[1] & \hcBN[1] \\ : & : \\ \hcBM[m] & \hcBN[m] \end{bsmallmatrix},\]
    where
    \[\hcA_k = 0, \, \hcBM[k] = U_k\trans, \; \hcBN[k] = \Sigma_k V_k\trans\]
    when \(\omega_k = 0\), and
    \[\hcA_k = \begin{bsmallmatrix}0 & \omega_k I \\ -\omega_k I & 0\end{bsmallmatrix} \; \text{and} \; \begin{bsmallmatrix}\hcBM[k] & \hcBN[k]\end{bsmallmatrix} = \begin{bsmallmatrix}\Re(U_k\trans) & \Re(\Sigma_k V_k\trans) \\ \Im(U_k\trans) & \Im(\Sigma_k V_k\trans)\end{bsmallmatrix}\]
    otherwise.  Construct the real-valued system 
    \[\widehat{R} := \brac{\begin{array}{c|c}\hcA-\hcBM\widehat{\bbW} & \hcBN - \hcBM D \\ \hline \widehat{\bbW} & D \end{array}}\]
    with real-valued free parameter \(\widehat{\bbW}\in\real^{p\times r}\).  If \(\bbW\) is defined as \(\bbW := \widehat{\bbW} J\), where the transformation matrix \(J\) is defined as
    \begin{align} \label{eq:J}
        J &:= \mathrm{diag}\{J_1, \ldots, J_m\}, \\
        J_k &:= \begin{cases}I & \text{if} \; \omega_k = 0 \\ \frac{1}{2}\begin{bsmallmatrix}I & I \\ jI & -jI\end{bsmallmatrix} & \text{otherwise}\end{cases}. \nonumber
    \end{align}
    Then, \(\widehat{R} = R\), i.e. there exists a similarity transformation between \(R\) and \(\widehat{R}\) with the transformation matrix \(J\).
\end{theorem}
\begin{proof}
    Consider interpolating at conjugate frequencies \(\pm\jw_k\), where \(\omega_k > 0\).  The interpolation data associated with the frequency \(\omega_k\) that appears \(\cBM\) and \(\cBN\) is \(U_k^*\) and \(\Sigma_k V_k^*\) respectively.  Considering \(G\) has real coefficients, it is the case that \(\overline{G}(\jw) = G(-\jw)\) for all \(\omega > 0\).  Given our interpolation condition that \(U_k^* G(\jw_k) = \Sigma_k V_k^*\), we see \(U_k\trans G(-\jw_k) = \Sigma_k V_k\trans\) from taking the conjugate, thus the interpolation data associated with \(-\omega_k\) are \(U_k\trans\) and \(\Sigma_k V_k\trans\).   Note that in the case that \(\omega_k = 0\), \(G(0)\) is real, thus \(U_k\), \(\Sigma_k\), and \(V_k\) are already real.  

    Constructing the system \(R\) as in the theorem statement and applying a coordinate transformation with the matrix \(J\) yields
    \[R = \brac{\begin{array}{c|c}J\cA J\inv - J\cBM\bbW J\inv & J\cBN - J\cBM D \\ \hline \bbW J\inv & D \end{array}}.\]
    Given that \(J\) is block diagonal, \(J\cA J\inv = \mathrm{diag}\curly{J_k \cA_k J_k\inv}\) and \(\begin{bsmallmatrix}J\cBM & J\cBN\end{bsmallmatrix} = \begin{bsmallmatrix}: & : \\ J\cBM[k] & J\cBN[k] \\ : & :\end{bsmallmatrix}\).  Noting that \(J_k\inv = \begin{bsmallmatrix}I & -jI \\ I & \phantom{-}jI\end{bsmallmatrix}\) when \(\omega_k > 0\), rote calculations provide \(J_k \cA_k J_k\inv = \hcA_k\), \(J\cBM[k] = \hcBM[k]\) and \(J\cBN[k] = \hcBN[k]\), thus
    \[J\cA J\inv = \hcA, \; J\cBM = \hcBM, \; \text{and} \; J\cBN = \hcBN.\]  
    If we define the complex-valued weights \(\bbW\) as a linear function of the real-valued weights \(\bbW = \widehat{\bbW} J\) and apply all aforementioned substitutions to the coordinate-transformed \(R\), we get exactly the realization \(\widehat{R}\), thus \(R=\widehat{R}\).  
\end{proof}

\begin{corollary} \label{cor:W_conj}
    The condition \(\bbW = \widehat{\bbW} J\) in Theorem~\ref{thm:real_R} enforces the weights associated with conjugate pairs of frequencies are themselves conjugate pairs.  
\end{corollary}
\begin{proof}
    Let \(\begin{bsmallmatrix}W_{k,1} & W_{k,2}\end{bsmallmatrix}\) and \(\begin{bsmallmatrix}\widehat{W}_{k,1} & \widehat{W}_{k,2}\end{bsmallmatrix}\) be the weight matrices associated with a conjugate pair of frequencies \(\pm \jw_k\) in \(R\) and \(\widehat{R}\) respectively.  From the condition \(\bbW = \widehat{\bbW} J\), we see 
    \begin{align*}
        \begin{bsmallmatrix}W_{k,1} & W_{k,2}\end{bsmallmatrix} &= \frac{1}{2}\begin{bsmallmatrix}\widehat{W}_{k,1} & \widehat{W}_{k,2}\end{bsmallmatrix} \begin{bsmallmatrix}I & I \\ jI & -jI\end{bsmallmatrix} \\
        \implies W_{k,1} &= \frac{1}{2}(\widehat{W}_{k,1} + j\widehat{W}_{k,2}) \\
        W_{k,2} &= \frac{1}{2}(\widehat{W}_{k,1} - j\widehat{W}_{k,2})
    \end{align*}
    Given that both \(\widehat{W}_{k,1}\) and \(\widehat{W}_{k,2}\) are real, it follows that \(W_{k,2} = \overline{W}_{k,1}\).
\end{proof}

Considering the above Theorem and Corollary, we see that given the real-valued free parameter \(\widehat{\bbW}\), we can obtain a real-valued interpolating system by either forming the matrices \(\hcA\), \(\hcBM\), and \(\hcBN\) and constructing \(\widehat{R}\) from Theorem~\ref{thm:real_R}, or applying a coordinate transformation to \(R\) from Theorem~\ref{thm:R_ss} using the matrix \(J\) from Theorem~\ref{thm:real_R}.

\subsection{The system \(H\)}
The system \(H := \mathcal{N} - \mathcal{M}G\) appears in Problem~\ref{pr:H2}, and we show in the next section the well-posedness depends on the properties of this system.  Here, we show that \(H\) admits a very nice representation in that its \(A\) and \(B\) matrices are identical to that of the input system. 

\newcommand*{\thmHrealstatement}{Consider the system \(G\) as in Notation~\ref{not:G} and matrices in Notation~\ref{amp:int_params}, and define the system \(H = \mathcal{N} - \mathcal{M}G\), where 
    \[\mathcal{M} = \begin{bsmallmatrix}I \\ \mathcal{M}_1(s) \\ : \\ \mathcal{M}_\ell(s)\end{bsmallmatrix}, \; \mathcal{N} = \begin{bsmallmatrix}D \\ \mathcal{N}_1(s) \\ : \\ \mathcal{N}_\ell(s)\end{bsmallmatrix}, \; \text{and}\]
    \[\begin{bsmallmatrix}\mathcal{M}_k & \mathcal{N}_k\end{bsmallmatrix} = \brac{\begin{array}{c|cc}j\omega_k I & U_k^* & \Sigma_k V_k^* \\ \hline I & 0 & 0\end{array}}.\]
    Then the system \(H\) has the realization 
    \begin{equation} \safelabel{eq:H_real}
        H = \brac{\begin{array}{c|c}A & B \\ \hline -C & 0 \\ \phantom{-}\cC & 0\end{array}}, \; \cC = \begin{bsmallmatrix}U_1^* C(\jw_1 I - A)\inv \\ : \\ U_\ell^* C(\jw_\ell I - A)\inv\end{bsmallmatrix} \in \comp^{r\times n}.
    \end{equation}}
\begin{theorem} \label{thm:H_real}
    \thmHrealstatement 
\end{theorem}

\begin{proof}
    Proof in the Appendix.
\end{proof}

\section{Weight Optimization}										\label{weights.sec}
Using the interpolating systems that were introduced in the previous section, we now turn to the optimization step of the algorithm, in which we assume a valid set of interpolation points has been chosen and we optimize a proxy for the error over all frequencies.  Here we derive an explicit solution for the weights that minimize a weighted \(L_2\) norm of the error system.  We also provide bounds on said \(L_2\) norm and show convergence to the input system when the number of states reaches the number of jointly observable and controllable states in the input system.

Because adjusting the weights affects the frequency response across the frequency spectrum, careful selection of the weights can be used to decrease the approximation error. We would like to choose the weights \(\bbW\) to minimize the mean squared error between \(R\) and \(G\), and ideally, we would try to optimize the \(L_2\) norm of \(R-G\).  This is a difficult problem due to the appearance of \(\bbW\) in both \(M\inv\) and \(N\), however, so instead we will focus on choosing \(\bbW\) such that it solves the weighted \(L_2\) minimization
\begin{align*}
    \min_{\bbW} \norm{M[\bbW](R[\bbW]-G)}_2^2 &= \min_{\bbW} \norm{N[\bbW]-M[\bbW]G}_2^2  \\
    &= \min_{\bbW} \norm{\begin{bsmallmatrix}I & \bbW\end{bsmallmatrix}(\smash{\underbrace{\mathcal{N} - \mathcal{M}G}_{H}})}_2^2,
\end{align*}
i.e. Problem~\ref{pr:H2}.  We later argue through computational examples that solving this optimization problem instead yields a system which still approximates the original system well in the \(L_2\) sense.  With some algebra, we see that this problem is equivalent to
\begin{align*}
    &\min_{\bbW} \frac{1}{2\pi}\int_{-\infty}^{\infty} \tr\p{\begin{bsmallmatrix}I & \bbW\end{bsmallmatrix} H(\jw)H^*(\jw) \begin{bsmallmatrix}I \\ \bbW^*\end{bsmallmatrix}} \rmd\omega \\
    =&\min_{\bbW} \tr\p{\begin{bsmallmatrix}I & \bbW\end{bsmallmatrix} X \begin{bsmallmatrix}I \\ \bbW^*\end{bsmallmatrix}},
\end{align*}
where 
\begin{equation} \label{eq:min1}
    X = \frac{1}{2\pi}\int_{-\infty}^{\infty} H(\jw) H^*(\jw) \rmd\omega. 
\end{equation}
For \(X\) to exist, \(H\) must have a zero feedthrough and \(H\) must not have any poles on the imaginary axis.  As shown in the previous section in Theorem~\ref{thm:H_real}, we see that \(H\) has a realization with zero feedthrough and has an \(A\) matrix which is identical to the \(A\) matrix of the input system \(G\).  Since we assume \(G\) has no poles on the imaginary axis, \(X\) exists, and a closed form for it can be constructed in terms of the controllability Gramian of the system \(G\).  Expanding \(H(\jw)\) with the new realization of \(H\) in equation~\eqref{eq:H_real}, we get 
\begin{align*}
    X &= \frac{1}{2\pi}\int_{-\infty}^{\infty} \begin{bsmallmatrix}-C \\ \phantom{-}\cC \end{bsmallmatrix} (\jw I-A)\inv BB^* (\jw I-A)^{-*} \begin{bsmallmatrix}-C \\ \phantom{-}\cC \end{bsmallmatrix}^* \rmd \omega \\
    &= \begin{bsmallmatrix}-C \\ \phantom{-}\cC \end{bsmallmatrix} \p{\frac{1}{2\pi}\int_{-\infty}^{\infty} \theta(\jw) \theta^*(\jw) \rmd \omega} \begin{bsmallmatrix}-C^* & \cC^*\end{bsmallmatrix},
\end{align*}
where \(\theta(s) = (sI-A)\inv B\).  This integral defines the matrix \(\Theta := \frac{1}{2\pi}\int_{-\infty}^{\infty}\theta(\jw)\theta^*(\jw)\,\rmd\omega\), which exists because \(G\) has no poles on the imaginary axis, and which equals the controllability Gramian of \(G\), i.e. the solution of \(A\Theta + \Theta A^* = -BB^*\), whenever \(G\) is stable.  Therefore
\[X = \begin{bsmallmatrix}C\Theta C^* & -C\Theta \cC^* \\ -\cC\Theta C^* & \cC \Theta \cC^*\end{bsmallmatrix},\]
and our minimization problem becomes 
\begin{equation*}
    \begin{split}
        &\min_{\bbW} \gamma(\bbW), \; \text{where} \\
        &\gamma(\bbW) = \tr \p{\begin{bsmallmatrix}I & \bbW\end{bsmallmatrix}\begin{bsmallmatrix}C\Theta C^* & -C\Theta \cC^* \\ -\cC\Theta C^* & \cC \Theta \cC^*\end{bsmallmatrix} \begin{bsmallmatrix}I \\ \bbW^*\end{bsmallmatrix}},
    \end{split}
\end{equation*}
which we have previously denoted Problem~\ref{pr:tr}.  According to Theorem~\ref{thm:What_sol}, this problem always has a minimizer, and the minimizer is unique exactly when \(\cC\Theta\cC^*\) is invertible.  The remaining results of this section are stated under this condition, which we record as an assumption.
\begin{assumption} \label{ass:CTC}
    The matrix \(\cC\Theta\cC^*\) is invertible.
\end{assumption}
By Lemma~\ref{lem:min_real}, \(\rank(\cC\Theta\cC^*) \leq n_{\min}\), the state dimension of a minimal realization of \(G\), so Assumption~\ref{ass:CTC} implies \(r \leq n_{\min}\).  Assumption~\ref{ass:CTC} holds automatically when \(p = 1\) (Lemma~\ref{lem:Ch_full}).  For \(p > 1\) the data must be selected to ensure this matrix is non-singular, which we discuss further in Section~\ref{inter_points.sec}.

\begin{theorem} \label{thm:What_sol}
    Consider the minimization in Problem~\ref{pr:tr}, using the expressions defined in Notation~\ref{not:G} and Notation~\ref{amp:int_params} when necessary, with \(\cC\) defined as in equation~\eqref{eq:H_real}.  The minimizers of \(\gamma\) are exactly the solutions \(\bbW\) of the linear equation
    \[\bbW\cC\Theta\cC^* = C\Theta\cC^*,\]
    which has at least one solution that is unique if and only if \(\cC\Theta\cC^*\) is invertible (Assumption~\ref{ass:CTC}).  In this case, 
    \[\bbW = C\Theta\cC^*(\cC\Theta\cC^*)\inv.\]
\end{theorem}
\begin{proof}
    Expand \(\gamma(\bbW)\) to get 
    \[\gamma(\bbW) = \tr\p{\bbW \cC \Theta \cC^* \bbW^* - \bbW \cC \Theta C^* - C\Theta \cC^* \bbW^* + C\Theta C^*}.\]
    Now, we calculate the Gateaux derivative of \(\gamma\), which yields
    \begin{align*}
        \partial_{\bbW}\gamma(V) &= \D{}{\delta} \gamma (\bbW + \delta V) \Big|_{\delta = 0} \\
        &= \tr(V\cC \Theta \cC^* \bbW^* + \bbW \cC \Theta \cC^* V^* \\
        &- V\cC \Theta C^* - C \Theta \cC^* V^*) \\
        &= 2\Re\brac{\tr(\bbW \cC \Theta \cC^* V^* - C \Theta \cC^* V^*)}.
    \end{align*}
    Because \(\cC\Theta\cC^*\) is positive semi-definite, \(\gamma\) is a convex quadratic function of \(\bbW\) whose minimizers are exactly the stationary points that occur when \(\partial_{\bbW} \gamma(V) = 0\) for all \(V\).  Thus the minimizers satisfy
    \begin{align*}
        \Re\brac{\tr(\bbW \cC \Theta \cC^* V^* - C \Theta \cC^* V^*)} = 0 \; \text{for all } V \\
        \implies \bbW \cC \Theta \cC^* = C \Theta \cC^*.
    \end{align*}
    This equation is always solvable.  Indeed, taking conjugate transposes gives \((\cC\Theta\cC^*)\bbW^* = \cC\Theta C^*\), and writing \(\Theta = LL^*\) shows that \(\range(\cC\Theta C^*) \subseteq \range(\cC L) = \range(\cC LL^*\cC^*) = \range(\cC\Theta\cC^*)\).  The solution is unique if and only if \(\cC\Theta\cC^*\) is invertible, in which case it is \(\bbW = C\Theta\cC^*(\cC\Theta\cC^*)\inv\).
\end{proof}

\subsection{The real-valued case}
When we require a real-valued realization as discussed in Theorem~\ref{thm:real_R}, \(\bbW\) is not a free parameter itself, but rather a function of a real-valued free parameter \(\widehat{\bbW}\), i.e. \(\bbW = \widehat{\bbW} J\), where \(J\) is a known matrix defined in equation~\eqref{eq:J}.  We may substitute this into Problem~\ref{pr:tr} to get the modified problem
\[\min_{\widehat{\bbW}} \tr\p{\begin{bsmallmatrix}I & \widehat{\bbW}\end{bsmallmatrix} \begin{bsmallmatrix}C\Theta C^* & -C\Theta\cC^* J^* \\ -J\cC \Theta C^* & J\cC \Theta \cC^* J^*\end{bsmallmatrix} \begin{bsmallmatrix}I \\ \widehat{\bbW}^*\end{bsmallmatrix}}.\]
Looking closely, it is clear that this is nearly identical to Problem~\ref{pr:tr}, with the only change being the matrix \(J\cC\) replacing \(\cC\).  Since \(J\) is invertible, \(J\cC\Theta\cC^*J^*\) is invertible exactly when Assumption~\ref{ass:CTC} holds, and Theorem~\ref{thm:What_sol} then gives the unique solution
\[\widehat{\bbW} = C\Theta \cC^* (\cC \Theta \cC^*)\inv J\inv.\]
Importantly, right multiplying by \(J\) to find \(\bbW\) reveals that it takes the exact form as it does in Theorem~\ref{thm:What_sol}.

\subsection{Stability and regularity}
With our new expression for the weight matrix \(\bbW\), we can now revisit the regularity condition laid out in Definition~\ref{def:reg_cond}, namely that \(R(s)\) is finite at the interpolation points.  We first find an alternate expression for the ``\(A\)'' matrix of \(R\) in Lemma~\ref{lem:R_A}.  

\begin{lemma} \label{lem:R_A}
    Consider the ``\(A\)'' matrix of \(R\) from equation~\eqref{eq:R_ss} in Theorem~\ref{thm:R_ss}, which is \(\cA - \cBM\bbW\).  Under Assumption~\ref{ass:CTC}, with \(\bbW\) the minimizer from Theorem~\ref{thm:What_sol},
    \[\cA - \cBM\bbW = \cC A \Theta \cC^*(\cC\Theta\cC^*)\inv.\]
\end{lemma}
\bigskip
\begin{proof}
    Substituting the optimal \(\bbW\) from Theorem~\ref{thm:What_sol} gives 
    \[\cA - \cBM C\Theta\cC^*(\cC\Theta\cC^*)\inv = (\underbrace{\cA\cC - \cBM C}_Y)\Theta\cC^*(\cC\Theta\cC^*)\inv.\]
    \vspace{-1em}
    
    \noindent
    Examining a block row of \(Y\) shows it has the following structure:
    \begin{align*}
        Y_k &= \jw_k U_k^*C(\jw_k I - A)\inv - U_k^* C \\
        &= (\cancel{\jw_k U_k^*C} - \cancel{\jw_k U_k^* C} + U_k^*CA)(\jw_k I - A)\inv \\
        &= U_k^*C(\jw_kI - A)\inv A = \cC_k A,
    \end{align*}
    due to the fact that \(A\) and \((\jw_k I - A)\inv\) commute, thus \(Y = \cC A\) and the claim in the Lemma statement immediately follows.   
\end{proof}

Using this simplified expression for the ``\(A\)'' matrix, we now show in Theorem~\ref{thm:R_nojwpoles} that, with the optimal weights, \(R\) is asymptotically stable whenever \(G\) is.  In particular, \(R\) has no purely imaginary poles, so the regularity condition in Definition~\ref{def:reg_cond} is satisfied without any further assumption on the interpolation data.  

\begin{theorem} \label{thm:R_nojwpoles}
    Let \(G\) be stable, let Assumption~\ref{ass:CTC} hold, and let \(\bbW\) be the minimizer from Theorem~\ref{thm:What_sol}.  Then every eigenvalue of the ``\(A\)'' matrix of the realization~\eqref{eq:R_ss} of \(R\) lies in the closed left half-plane, and every eigenvalue on the imaginary axis is uncontrollable and is therefore not a pole of \(R(s)\).  Consequently, \(R(s)\) is asymptotically stable, is finite at every \(\jw_k\), and satisfies the regularity condition in Definition~\ref{def:reg_cond}.
\end{theorem}

\begin{proof}
    By Lemma~\ref{lem:R_A}, the ``\(A\)'' matrix of~\eqref{eq:R_ss} is \(A_r := \cC A\Theta\cC^*(\cC\Theta\cC^*)\inv\), and the ``\(B\)'' matrix of~\eqref{eq:R_ss} is \(\cC B\), since its \(k\)th block row is \(\Sigma_k V_k^* - U_k^* D = U_k^*(G(\jw_k) - D) = U_k^* C(\jw_k I - A)\inv B = \cC_k B\).  Right-multiplying \(A_r\) by \(\cC\Theta\cC^*\), adding the conjugate transpose, and substituting the Lyapunov equation \(A\Theta + \Theta A^* = -BB^*\), which holds because \(G\) is stable, gives
    \begin{align}
       &A_r(\cC\Theta\cC^*) + (\cC\Theta\cC^*)A_r^* \nonumber \\
       =&\cC(A\Theta + \Theta A^*)\cC^* = -\cC BB^*\cC^*, \label{eq:Ar_lyap}
    \end{align}
    noting that the matrix \(\cC\Theta\cC^*\) is positive semi-definite and invertible, hence positive definite.

    Let \(\lambda\) be an eigenvalue of \(A_r\) with left eigenvector \(w \neq 0\), i.e. \(w^* A_r = \lambda w^*\).  Left-multiplying~\eqref{eq:Ar_lyap} by \(w^*\) and right-multiplying by \(w\) yields
    \[2\,\Re(\lambda)\, w^*\cC\Theta\cC^* w = -\norm{B^*\cC^* w}^2,\]
    and since \(w^*\cC\Theta\cC^* w > 0\), it follows that \(\Re(\lambda) \leq 0\).

    Now suppose \(\Re(\lambda) = 0\).  Then the above equation gives \(w^*\cC B = 0\) for every left eigenvector \(w\) of \(\lambda\).  Furthermore, since~\eqref{eq:Ar_lyap} has a positive definite solution, \(A_r\) is Lyapunov stable and \(\lambda\) is a semisimple eigenvalue.  Performing a coordinate change to a Jordan basis of \(A_r\) then shows that the block row of the ``\(B\)'' matrix corresponding to the \(\lambda\) block consists of rows of the form \(w^*\cC B\) with \(w\) a left eigenvector, and is therefore zero; thus the \(\lambda\) modes are uncontrollable and \(\lambda\) is not a pole of \(R(s)\).  
    
    In summary, every eigenvalue of \(A_r\) has non-positive real part, and those with zero real part are not poles of \(R(s)\).  Hence every pole of \(R(s)\) has strictly negative real part, i.e. \(R(s)\) is asymptotically stable.  In particular, \(R(\jw_k)\) is finite for every \(k\), which is the regularity condition in Definition~\ref{def:reg_cond}.
\end{proof}

\begin{remark}[Relation to~\cite{Gugercin2008}] \label{rem:projection}
    Under the assumptions of Theorem~\ref{thm:R_nojwpoles}, let \(V := \Theta\cC^*(\cC\Theta\cC^*)\inv\) and \(Z := \cC^*\).  Then \(Z^*V = I\), and the realization~\eqref{eq:R_ss} with the optimal weights is 
    \[\brac{\begin{array}{c|c} Z^*AV & Z^*B \\ \hline CV & D\end{array}}.\]
    This is a Petrov--Galerkin projection whose test space is the left-tangential rational Krylov subspace of the interpolation data and whose trial space is its image under \(\Theta\).  For \(p = 1\) this coincides with the Gramian-weighted projection of~\cite[Remark~3.4]{Gugercin2008}, however the derivations differ.   In~\cite{Gugercin2008} the Gramian side is imposed to obtain stability of the resulting model, and the fixed set of interpolation points is relocated to the mirror images of the reduced poles a number of times until \(H_2\) optimality is reached.  Here the Gramian arises naturally from the weighted \(L_2\) problem of Theorem~\ref{thm:What_sol} and Theorem~\ref{thm:R_nojwpoles} gives stability for arbitrary data, which is adaptively selected in this work.  Since the test space contains \((\jw_k I - A)^{-*}C^*U_k\), left interpolation at \(\jw_k\) also follows from the standard projection argument~\cite{Gallivan04,Antoulas2020} whenever \(\jw_k\) is not an eigenvalue of \(Z^*AV\), so the optimal weights need not satisfy the full-column-rank hypothesis of Theorem~\ref{thm:tan_int}.
\end{remark}

\subsection{Performance across iterations}
Now that we have a closed-form solution to the optimization problem, and have clarified the satisfiability of the regularity condition using this closed-form solution, we want to understand how the performance of the algorithm across iterations is determined by this choice on \(\bbW\).  The optimization problem minimizes the \(L_2\) norm of \(N - MG\), which is equal to \(\gamma^{(i)}(\bbW^{(i)})\) at iteration \(i\).  Thus, understanding how \(\gamma\) behaves across iterations gives insight to the overall performance of the algorithm.  Theorem~\ref{thm:norm_behavior} answers a few of these questions.  We see that \(\gamma\) is bounded above by \(\norm{G-D}_2\) and that it is non-increasing across iterations.  Under Assumption~\ref{ass:CTC}, it reaches zero exactly when the size of \(R\) equals the number of jointly controllable and observable modes of the original system, in which case \(R = G\) (up to coordinate changes).  

Though this does not shed any light on the true approximation error \(\norm{R-G}\) across iterations, we show in the numerical section that these values tend to be tightly coupled, suggesting the true error norm may be bounded above by an expression involving \(\gamma\). 
\begin{theorem} \label{thm:norm_behavior}
    Consider the system \(G\) defined as in Notation~\ref{not:G}, let \(\bbW^{+(i)}\) be a minimizer of \(\gamma^{(i)}\) at iteration \(i\) from Problem~\ref{pr:tr}, and let \(\gamma^{+(i)} := \gamma^{(i)}(\bbW^{+(i)})\).  Then \(\gamma^{+(0)} = \norm{G-D}_2^2\) and, for \(i \geq 0\),
    \[\gamma^{+(i+1)} \leq \gamma^{+(i)}.\]
    Furthermore, under Assumption~\ref{ass:CTC}, \(\gamma^{+(i)} = 0\) when \(r = n_{\min}\).
\end{theorem}
\begin{proof}
    When \(i=0\), \(\cC\) is an empty \(0\times n\) matrix, and \(\gamma^{+(0)}\) simplifies to \(\gamma^{+(0)} = \tr(C\Theta C^*)\), which by definition is the square of the \(L_2\) norm of the \(G\) without a \(D\) matrix.  

    For any iteration \(i\), let \(\bbW^{+(i)}\) be a minimizer of \(\gamma^{(i)}\), define the \textit{next iteration's} matrix \(\cC^{(i+1)} = \begin{bsmallmatrix}\cC^{(i)} \\ *\end{bsmallmatrix}\), and consider letting \(\bbW^{(i+1)} = \begin{bsmallmatrix}\bbW^{+(i)} & 0\end{bsmallmatrix}\).  Evaluating \(\gamma^{(i+1)}\), we get
    \[\gamma^{(i+1)}(\bbW^{(i+1)}) = \tr\p{(C - \bbW^{+(i)}\cC^{(i)})\Theta (C - \bbW^{+(i)}\cC^{(i)})^*},\]
    which is exactly equal to \(\gamma^{+(i)}\).  Since \(\bbW^{(i+1)}\) is a feasible choice of weights at iteration \(i+1\), the minimal value cannot exceed it, i.e. \(\gamma^{+(i+1)} \leq \gamma^{+(i)}\).

    Finally, let Assumption~\ref{ass:CTC} hold, so that \(\bbW^{+(i)} = C\Theta\cC^{(i)*}(\cC^{(i)}\Theta\cC^{(i)*})\inv\), and suppose \(r = n_{\min}\).  By Lemma~\ref{lem:min_real}, \(C\Theta C^* = C_{\min}\Theta_{\min}C_{\min}^*\), \(C\Theta\cC^{(i)*} = C_{\min}\Theta_{\min}\cC_{\min}^{(i)*}\), and \(\cC^{(i)}\Theta\cC^{(i)*} = \cC_{\min}^{(i)}\Theta_{\min}\cC_{\min}^{(i)*}\), where \(\cC_{\min}^{(i)}\in\comp^{n_{\min}\times n_{\min}}\) is square and, since the last product is invertible, itself invertible.  Substituting these into \(\bbW^{+(i)}\) gives \(\bbW^{+(i)} = C_{\min}\cC_{\min}^{(i)-1}\), i.e. \(\bbW^{+(i)}\cC_{\min}^{(i)} = C_{\min}\).  The same substitutions show \(\gamma^{+(i)} = 0\).
\end{proof}

\begin{corollary} \label{cor:RG_equal}
    Under Assumption~\ref{ass:CTC}, when \(r = n_{\min}\), \(\norm{R-G}_2 = 0\), thus \(R=G\).
\end{corollary}
\begin{proof}
    By Theorem~\ref{thm:norm_behavior}, \(\gamma^{+} = \norm{N - MG}_2^2 = 0\) when \(r = n_{\min}\), thus \(N(s) = M(s)G(s)\).  Since \(M(s)\) is invertible for all but finitely many \(s\) (its feedthrough term is \(I\)), left-multiplying by \(M\inv(s)\) gives \(R(s) = G(s)\), i.e. there exists a similarity transform between their minimal realizations.
\end{proof}
To summarize the above results, choosing the weight matrix \(\bbW\) via an optimization on the weighted norm \(\norm{N-MG}_2\) admits an explicit solution that has a simple form and is computationally efficient to compute.  Under Assumption~\ref{ass:CTC}, this choice of \(\bbW\) yields a stable reduced-order model whenever \(G\) is stable, so the regularity condition is satisfied without further conditions on the interpolation data.  Furthermore, we have shown said weighted norm is non-increasing across iterations, and once \(R\) is as large as a minimal realization of \(G\), it is equal to \(G\) up to coordinate changes.  

\section{Interpolation points selection}								\label{inter_points.sec}
The previous two sections laid out the theoretical foundation for the work we present in this paper.  In this section, we explore various strategies of selecting the interpolating frequencies as well as the approximation rank at said frequencies.  In the first subsection, we first introduce a strategy inspired by the AAA algorithm and interpolate at the frequency where the error is largest.  Following that, we consider a random strategy and a grid strategy, both of which are computationally simpler.  We also lay out some important considerations in terms of choosing the rank of the approximation for each candidate frequency we select.  We end the section by summarizing the algorithm as a whole and provide pseudocode detailing the implementation.  

\subsection{Frequency selection strategies}
As detailed in the previous section, the weighted error norm will not increase across iterations, giving a lot of freedom in choosing the next interpolation frequency.  Intuitively, interpolating where the error is maximized should improve the approximation at that point and nearby points, thus lowering the total error more than any other choice.  This is the approach in the AAA algorithm, where the next interpolation point is chosen to the point which has the largest approximation error from the large set of test points~\cite{Nakatsukasa_2018}.  Under that reasoning, interpolating at the frequency where the spectral norm, denoted \(\norm{\cdot}\), of the error system \(R-G\) is maximized is a reasonable choice for our algorithm.  In other words, at the \(i\)th iteration, we select \(\omega_i\) to be 
\begin{equation}
    \omega_i = \arg\max_{\omega\in\real} \norm{R(\jw)-G(j\omega)}.
\end{equation}
This frequency occurs where the peak gain is attained for the system \(R-G\), thus an \(H_\infty\) norm computation may be used to calculate where the next interpolation point will be.  These methods are already widely available, such as in MATLAB, where one can use the \texttt{getPeakGain} function, which uses a variant of the algorithm in~\cite{Bruinsma90} and works on both stable and unstable systems.  This algorithm converges in a small number of steps, but at each of those steps, all eigenvalues of a Hamiltonian matrix with size \(2(r+n)\) must be computed, which can be quite slow for large systems computationally speaking; when \(r \ll n\), this step would have a computational complexity of \(O(n^3)\).

\begin{algorithm}
    \caption{Maximum error frequency selection}
    \label{alg:max_error}
    \begin{algorithmic}[1]
        \Require{\(G\), \(R\)}
        \State \Return frequency returned from \(\mathtt{getPeakGain}(G-R)\)
    \end{algorithmic}
\end{algorithm}

\subsubsection{Random and discrete strategies}
Unfortunately, interpolating at the maximum error may be unacceptably slow for large systems, leading us to consider two alternative approaches that sidestep this computational bottleneck.  The first strategy is a discrete approach where we are given a set of frequencies \(\curly{\hat{\omega}_k}_{k=1}^K\) and choose the interpolation point from this set where the error is highest, i.e. the maximum value of \(\norm{R(j\hat{\omega}_k) - G(j\hat{\omega}_k)}\) among all \(\hat{\omega}_k\).  Assuming a zero-pole-gain input model is available and \(r\ll n\), this operation may be completed in \(O(Kn)\).  The performance of this method may depend on how fine the grid is, so \(K\) may have to be very large to get good performance.  

\begin{algorithm}
    \caption{Discrete frequency selection}
    \label{alg:discrete}
    \begin{algorithmic}[1]
        \Require{\(G\), \(R\), \(\curly{\hat{\omega}_k}_{k=1}^K\)}
        \State \Return the \(\hat{\omega}_k\) that maximizes \(\norm{R(j\hat{\omega}_k) - G(j\hat{\omega}_k)}\)
    \end{algorithmic}
\end{algorithm}

We will consider an additional strategy of choosing the frequency stochastically.  First, we pick \(K\) random frequencies \(\curly{\hat{\omega}_k}_{k=1}^K\) logarithmically between \(\omega_{\min}\) and \(\omega_{\max}\).  Then, we calculate the error norm for each frequency and let \(\omega_i\) be the frequency that maximizes the norm.  Like the previous method, this operation may be computed in \(O(Kn)\).

\begin{algorithm}
    \caption{Random frequency selection}
    \label{alg:random}
    \begin{algorithmic}[1]
        \Require{\(\omega_{\min}\), \(\omega_{\max}\), \(R\), \(G\), \(K\)}
        \State \(\curly{\hat{\omega}_k}_{k=1}^K \gets\) \(K\) log-randomly distributed numbers in \([\omega_{\min}, \, \omega_{\max}]\)
        \State \Return the \(\hat{\omega}_k\) that maximizes \(\norm{R(j\hat{\omega}_k) - G(j\hat{\omega}_k)}\)
    \end{algorithmic}
\end{algorithm}

\subsection{Rank considerations}
The final consideration is the method of choosing the rank of interpolation at each frequency.  In most cases, interpolating the singular vector corresponding to the largest singular value is acceptable.  However, if the next largest singular values at that frequency are very close in magnitude, interpolating all of them is worth consideration.  We can then define a ``cut-off'' value \(0<\rho\leq 1\) and include all singular vectors whose singular values are at least \(\rho\) times the largest singular value not yet interpolated at that frequency.  A value close to 1 will have better performance and a smaller value will decrease the number of total iterations.  We use \(\rho=0.99\) as a default value in the code.  

In the case that the candidate frequency is close to an already-interpolated frequency, we may instead increase the approximation rank of the existing frequency by incorporating the next largest singular vector.  After calculating a candidate frequency \(\omega\), consider the closest existing interpolation frequency \(\omega_i\).  If the relative difference \(|(\omega_i-\omega)/\omega|\) is less than some threshold value \(\mu\), then we may increase the approximation rank of \(\omega_i\) instead of adding a new interpolation point.  In other words, we can add new columns to the existing \(U_i\) and \(V_i\) corresponding to the existing frequency \(\omega_i\).  A small value of \(\mu\) will preferentially add new interpolation data, and a large value will preferentially increase the approximation rank of nearby interpolation points.  It should be chosen small enough so that the closest two successive peaks in the frequency response would be treated separately.  The implementation of these two ideas is shown in Algorithm~\ref{alg:refine}.

\begin{algorithm}
    \caption{Frequency and rank refinement}
    \label{alg:refine}
    \begin{algorithmic}[1]
        \Require{\(G\), \(m\), \(\curly{\omega_k}_{k=1}^m\), \(\curly{r_k}_{k=1}^m\), \(\omega\), \(\mu\), \(\rho\)}
        \State \(r_{\min}, \, \mathrm{new} \gets 1, \, \mathrm{true}\)
        \If{\(m > 0\)}
            \State \(i \gets\) index of nearest freq. in \(\curly{\omega_k}_{k=1}^m\) to \(\omega\)
            \If{\(\omega_i = \omega = 0\) or \(\left|(\omega_i - \omega)/\omega\right| < \mu\)}
                \State \(\omega, \, r_{\min}, \, \mathrm{new} \gets \omega_i, \, r_i+1, \, \mathrm{false}\)
            \EndIf
        \EndIf
        \State \(\sigma \gets\) \(r_{\min}\)th largest singular value of \(G(j\omega)\)
        \State \(r_{\max}\gets\) index of smallest S.V. \(\geq \rho \sigma\)
        \If{new}
            \State \(m\gets m+1\)
            \State \(\omega_m, \; r_m \gets \omega, r_{\max}\)
        \Else
            \State \(r_i \gets r_{\max}\)
        \EndIf
        \State \Return \(\omega\), \(r_{\min}\), \(r_{\max}\)
    \end{algorithmic}
\end{algorithm}

\subsection{The algorithm}
Algorithm~\ref{alg:method} provides an overview of the algorithm.  First, various variables are initialized and the matrix \(\Theta\) of Section~\ref{weights.sec} is calculated, via the Lyapunov equation when \(G\) is stable.  Then, one of the previous frequency selection methods is used to find a candidate frequency, which is then used in the rank refinement step.  When Algorithm~\ref{alg:refine} increases the rank at an existing frequency instead, the new data is appended to that frequency's existing blocks.  A SVD of the chosen frequency's response is performed, and the resulting data is used in the solution of two linear equations.  Before the weights are computed, the algorithm checks that \(\cC\Theta\cC^*\) is well conditioned (the practical form of Assumption~\ref{ass:CTC}) and discards the candidate data if its condition number exceeds a threshold \(\tau\).  A rejected direction is one in which the current \(R\) already interpolates \(G\) at \(\jw\), by the same projection argument as in Remark~\ref{rem:projection}, so the next singular vector at \(\omega\) is tried and a new frequency only if none remain.  Finally, the reduced order system is constructed.  A full implementation of each algorithm and all related code may be found at the GitHub repository \url{https://github.com/bluesun212/ITFramework}.

\begin{algorithm}
    \caption{The low-rank tangential algorithm}
    \label{alg:method}
    \begin{algorithmic}[1]
        \Require{State space system \(G\), \(\mu\), \(\rho\), \(\tau\)}
        \Require{\(K\), \(\omega_{\min}\), \(\omega_{\max}\), \(\curly{\hat{\omega}_k}_{k=1}^K\) depending on algorithm}
        \State \(R \gets D\)
        \State Initialize empty \(\mathcal{A}, \, \cBM, \, \cBN, \, \cC\)
        \State \(\Theta \gets\) the matrix \(\Theta\) defined in Section~\ref{weights.sec}
        \Repeat
            \State Use one of Algorithms~\ref{alg:max_error}-\ref{alg:random} to get \(\omega\)
            \State Use Algorithm~\ref{alg:refine} to compute \(\omega, \, r_{\min}, \, r_{\max}\)
            \State \(U, \, \Sigma, \, V \gets\) the \(r_{\min}\)th to \(r_{\max}\)th singular values/vectors from the SVD of \(G(j\omega)\)
            \State Solve \(C_\omega (j\omega I - A) = U^*C\) for \(C_\omega\)
            \If{\(G\) is real-valued and \(R\) must also be}
                \State Append data to \(\cA\), \(\cBM\), and \(\cBN\) according to Theorem~\ref{thm:real_R}
                \State Append \(C_\omega\) and its conjugate if \(\omega\neq 0\) to \(\cC\)
            \Else
                \State Append data to \(\cA\), \(\cBM\), and \(\cBN\) according to Theorem~\ref{thm:R_ss}
                \State Append \(C_\omega\) as a row to \(\cC\)
            \EndIf
            \If{\(\mathrm{cond}(\cC\Theta\cC^*) > \tau\)}
                \State Discard this iteration's data; retry with the next singular vector at \(\omega\) (line 7) or a new frequency (line 5)
            \EndIf
            \State Solve \(\bbW \cC \Theta \cC^* = C \Theta \cC^*\) for \(\bbW\)
            \State \(R\gets \brac{\begin{array}{c|c}\mathcal{A}-\cBM\bbW & \cBN - \cBM D \\ \hline \bbW & D \end{array}}\)
        \Until{\(R\) sufficiently approximates \(G\)}
        \State If a real realization of \(R\) is required, apply a coordinate transformation with matrix \(J\) as detailed in Theorem~\ref{thm:real_R}
    \end{algorithmic}
\end{algorithm}

\section{Computational results}									\label{computations.sec}
We will now demonstrate the performance of our proposed algorithms, specifically Algorithm~\ref{alg:method} using the various frequency selection methods detailed in Algorithms~\ref{alg:max_error}-\ref{alg:random}, with numerical examples.  To quantify each algorithm's performance, we use the \(L_2\) norm of the error system \(R - G\) from Section~\ref{formulation.sec} (or an approximation thereof) as the indicator of approximation performance, which we denote ``error norm'' in the figures.  The code that generated all of the following figures may be found in the GitHub repository stated above.

In the following three figures, we use the so-called ``ISS'' model as the input system, which is a 3-input, 3-output, 270-state flexural model of one of the modules found on the International Space Station~\cite{iss_model}.  Fig.~\ref{fig:alg_comp} shows the error norm plotted against the reduced systems' state dimension for a number of algorithms, namely balanced truncation, the Loewner framework, tangential-AAA, and our proposed method using the maximum error frequency selection approach.  Up until about 20 states, each of the methods perform similarly.  After that, our proposed method sees continued improvement as the number of states increases, and continues to give similar results compared to balanced truncation.  We note that the performance of these Loewner methods depends heavily on the input grid spacing in order to accurately resolve the input system dynamics, a problem which our method avoids.

\begin{figure}[ht!]
    \centering
    \includegraphics[width=0.9\linewidth]{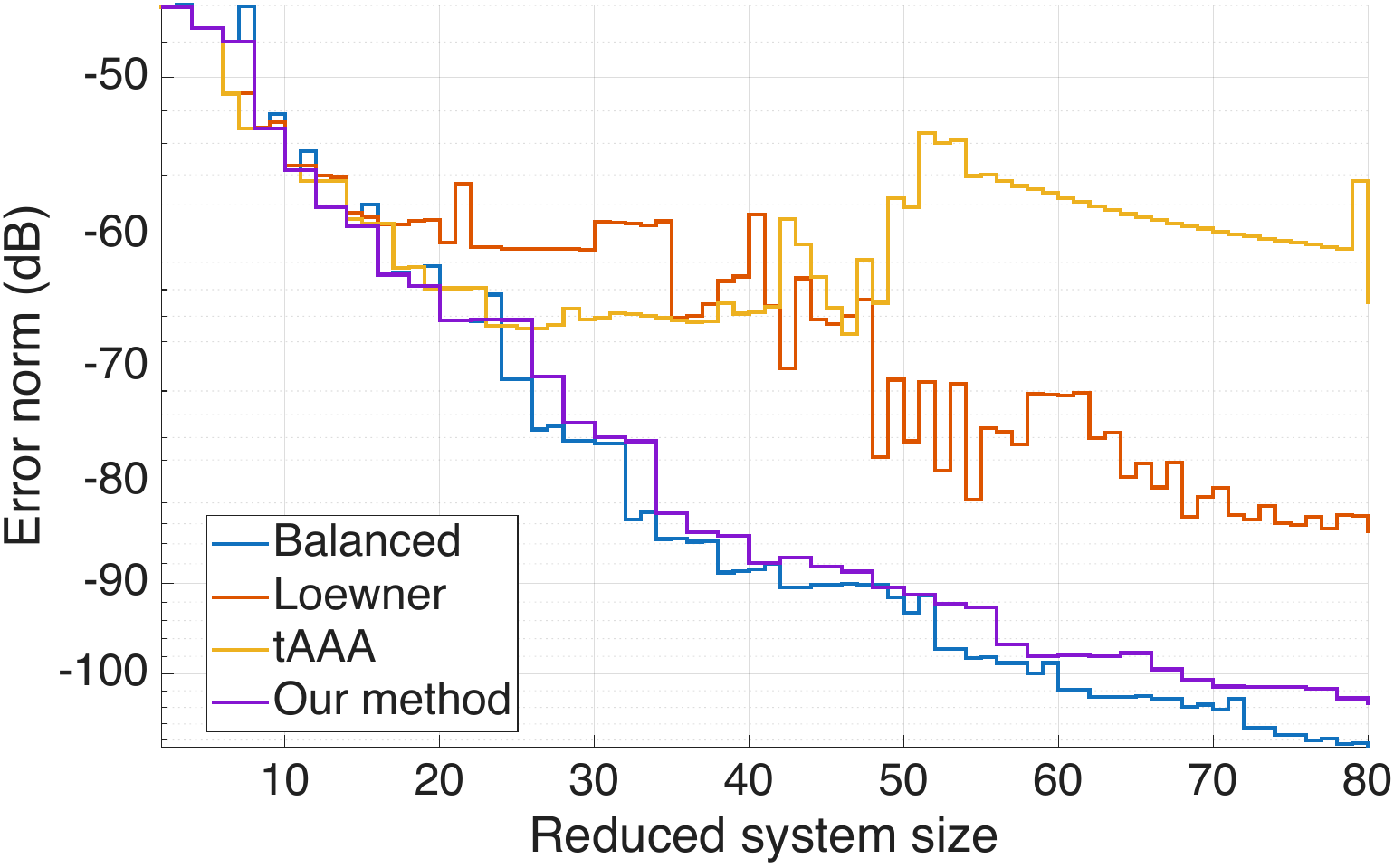}
    \mycaption{A comparison of the performance of balanced truncation in blue, the generalized Loewner framework~\cite{antoulas2017tutorial} in red, the tangential-AAA (tAAA) algorithm~\cite{Benner23} in yellow, and one of our proposed methods (max. error algorithm) in purple.  The tAAA algorithm and the Loewner framework both were provided an input set of 1000 points logarithmically spaced between \(\pm[10^{-1}, \, 10^2]j\).  The graph shows the norm of the error system versus the reduced system's number of poles.  The input model is the 3-input, 3-output, 270-state ``ISS'' model~\cite{iss_model}.} 
    \label{fig:alg_comp}
\end{figure}

We now compare the performance of each of our frequency selection methods, starting with the gridded strategy as shown in Fig.~\ref{fig:gridded}.  In each iteration of the gridded approach, an interpolation point is added at the frequency of maximum error among a grid of \(K=50\), 200, and 1000 logarithmically spaced points between \([10^{-1}, \, 10^2]\).  We see that performance is very good with \(K=200\) and nearly identical to that of the maximum error approach with \(K=1000\).  A grid with many points is required in order to accurately resolve fine details. 

\begin{figure}[ht!]
    \centering
    \includegraphics[width=0.9\linewidth]{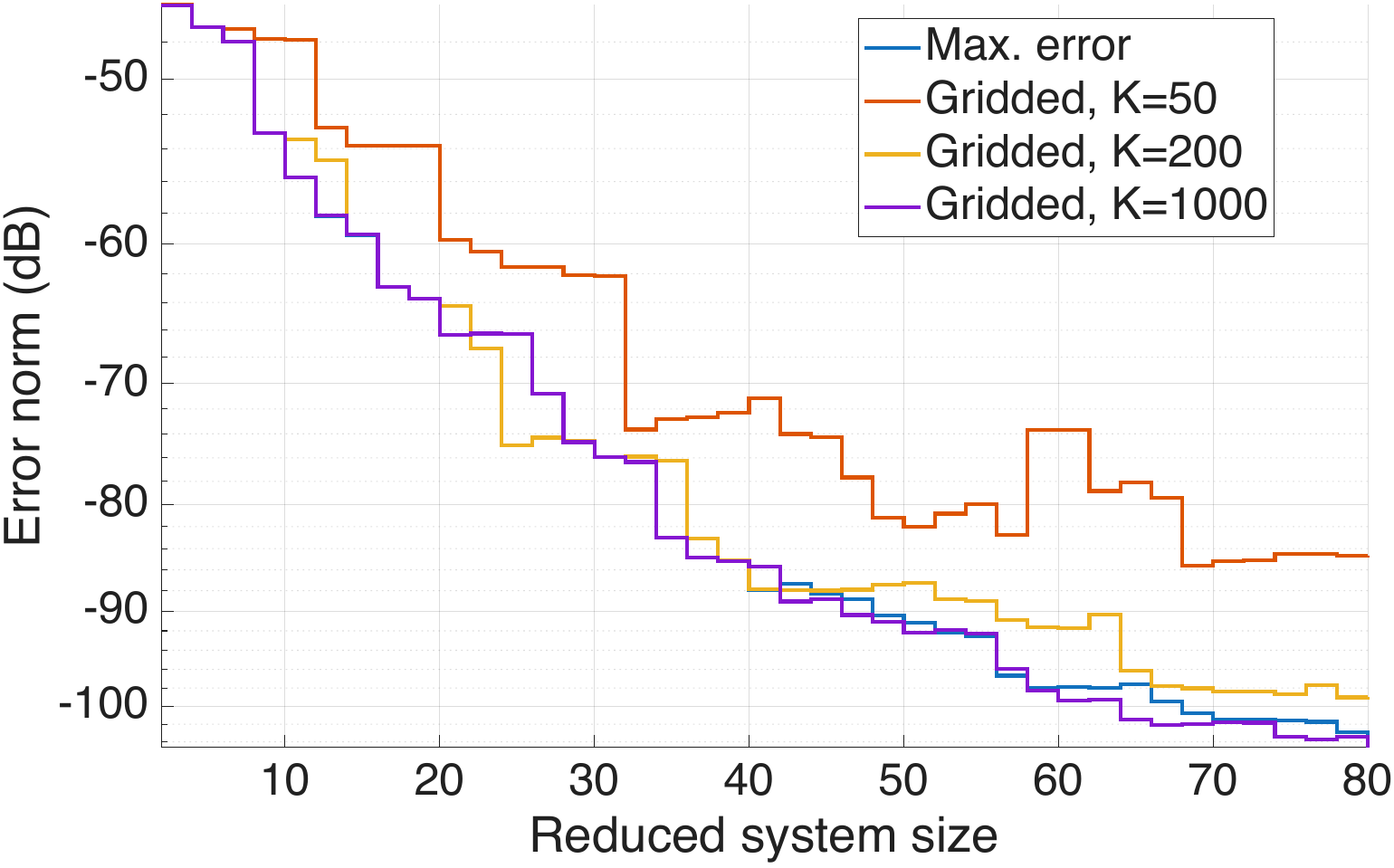}
    \mycaption{A comparison of the performance of the maximum error strategy (in blue) and the gridded frequency point selection (in varying colors) with varying \(K\), the number of points in the grid.}
    \label{fig:gridded}
\end{figure}

A comparison of the random strategy, which selects the frequency of highest error from \(K\) logarithmically sampled points between \([10^{-1}, \, 10^2]\), is shown in Fig.~\ref{fig:random}.  For \(K=5\), 20, and 100, we reduced the ISS system and recorded the error norm for 500 different random starting seeds.  The mean behavior of these trials and the middle 50\% performance is shown in Fig.~\ref{fig:random} in comparison to the maximum frequency approach.  We note as \(K\) increases, the average performance increases and the variance among runs decreases.  When \(K=20\), the mean behavior is markedly better than that of \(K=5\), and when \(K=100\), the mean performance practically identical to the maximum error approach.  

\begin{figure}[ht!]
    \centering
    \includegraphics[width=0.9\linewidth]{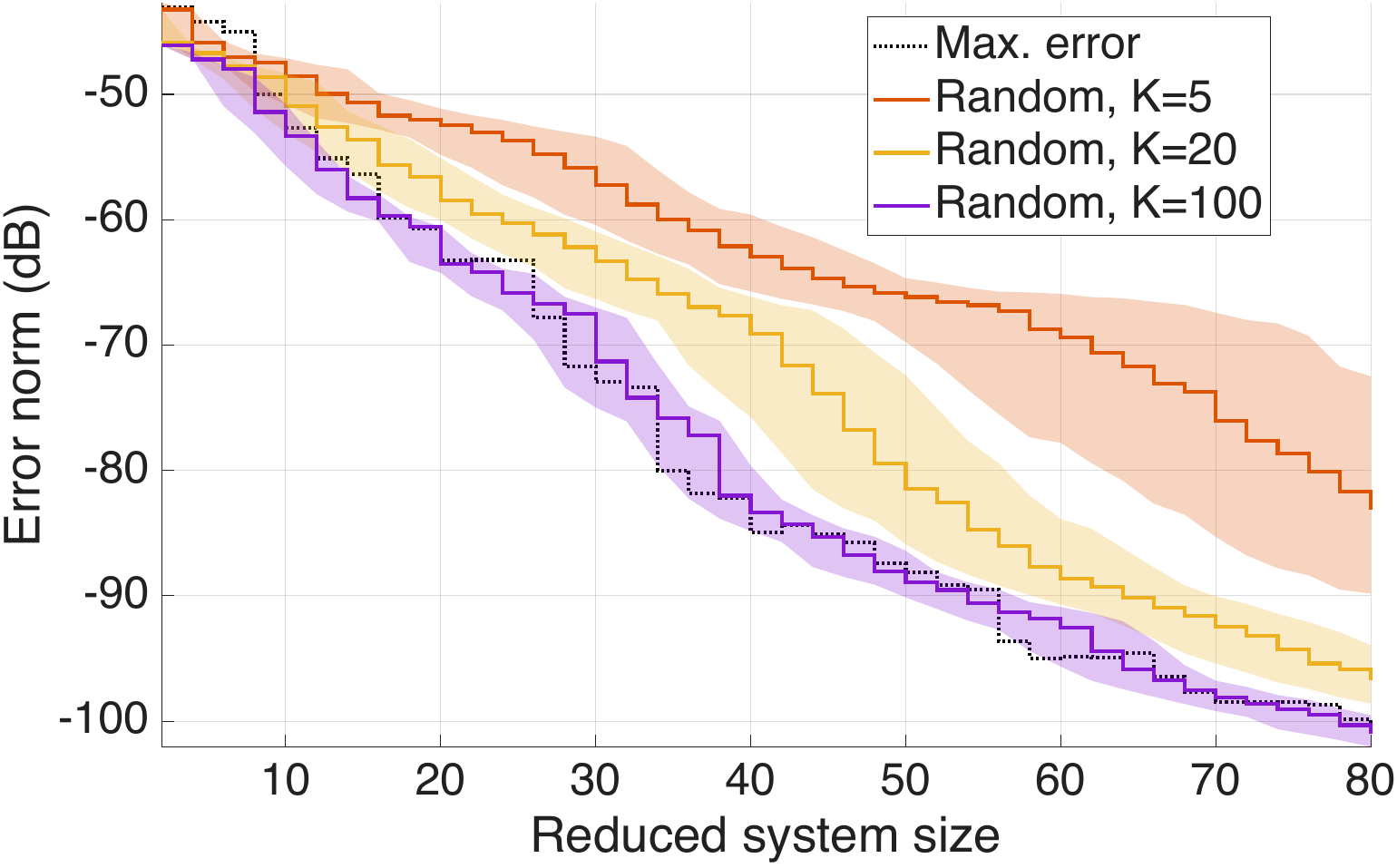}
    \mycaption{A comparison of the performance of the maximum error strategy (in black) and random frequency point selection strategy (in varying colors) with varying \(K\), the number of points tested per iteration, versus reduced system size.   The solid lines represent the mean behavior of a number of random runs, with the shaded regions representing the middle 50\% behavior.}
    \label{fig:random}
\end{figure}

We note that the mean random approach requires an order of magnitude fewer points than the gridded approach, but the performance is non-deterministic, which is an important trade-off that must be considered when selecting the most appropriate algorithm.  Nevertheless, these two figures indicate that good performance can be attained with much less computational power required compared to the maximum error approach.  They also show that the exact choice of the frequency is not so critical, due to the fact that the system dynamics are sufficiently captured in the weight optimization step.

Regarding stability, Theorem~\ref{thm:R_nojwpoles} guarantees that the reduced-order model is stable whenever the input model is, in contrast to moment matching methods in general.  In numerical experiments on randomly-generated stable systems, we observe unstable reduced models only when the system is nearly completely reduced (relative error \(\norm{R-G}_2/\norm{G}_2\) below \(10^{-6}\)) and the matrix \(\cC\Theta\cC^*\) is extremely ill-conditioned; these instabilities are numerical rather than structural.

Something important that we did not cover in the theoretical sections is how the true approximation error \(\norm{R-G}_2\) changes across iterations.  For 500 randomly-generated stable systems with 250 states, 3 inputs, and 3 outputs, we calculated the ratio of the true error norm over the non-increasing weighted error norm \(\sqrt{\gamma} = \norm{MG-N}_2\) and plotted the mean results as well as middle 50\% and 75\% behavior.  This ratio remains close to 1 for the majority of systems tested, showing the weighted error norm is a great indicator of the true approximation error.  

\begin{figure}[ht!]
    \centering
    \includegraphics[width=0.9\linewidth]{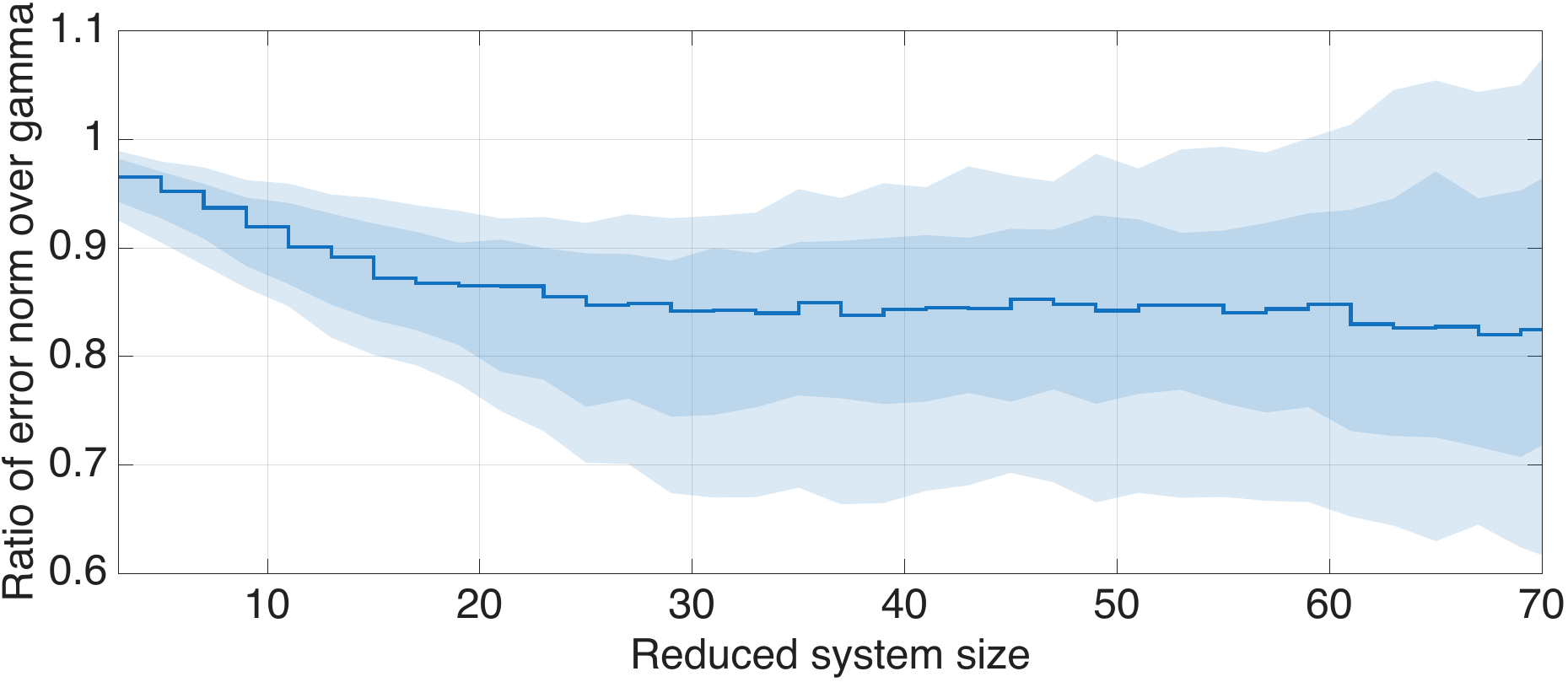}
    \mycaption{The ratio \(\scriptscriptstyle\frac{\norm{R-G}_2}{\norm{MG-N}_2}\) of the true error norm over the weighted error norm \(\sqrt{\gamma}\) versus reduced system size.  The solid line represent the mean behavior of 500 randomly-generated 250-state, 3 input, 3 output stable systems, and the shaded regions represent the middle 50\% and 75\% behaviors for darker and lighter respectively.}
    \label{fig:rss_norm}
\end{figure}

\section{Discussion and conclusion}
In this work we have introduced an algorithm based on the AAA algorithm and the tangential-AAA algorithm for the model reduction of MIMO state-space systems.  We derived a state-space representation of the tangential interpolant in both the real-valued and more general complex valued case, each of which are constructed from frequency response data and a free parameter ``weight'' matrix.  Furthermore, we characterized sufficient conditions for when the resulting reduced-order system satisfies a certain general interpolation condition.  The free parameter matrix is chosen as the solution of an optimization problem to minimize a weighted approximation error across the frequency axis, and has a closed form solution which can be found by solving a linear equation.  Importantly, under a single condition that our algorithm enforces when selecting the interpolation data, this choice of free parameters ensures that a weighted \(L_2\) norm of the error is non-increasing across iterations and reaches 0 when the resulting system's state dimension equals the number of observable and controllable modes of the original system, and that the reduced-order system is asymptotically stable whenever the original system is.  We discussed the method of interpolation frequency selection and introduced three algorithms, each of which trade off approximation performance with computational complexity.  The numerical examples demonstrated good performance on par with standard model reduction methods and improved performance compared to other moment matching methods.  

In future works we plan to adapt this algorithm for use on sparse systems and explore more sophisticated frequency selection methods with a focus on computational efficiency.  We also plan to prove more theoretical properties, including the existence of monotonically decreasing upper bounds for the \(L_2\) and \(L_\infty\) norms of the error system, and expressions describing the improvement across iterations. 

\appendices

\section{}
% \textit{Notation~\ref{amp:int_params}}: \ampintparamsstatement

\begin{lemma} \label{lem:reg_cond}
    Consider \(M(s)\) and \(N(s)\) from the interpolant in eq.~\eqref{eq:rtan_int}, which we write as
    \[M(s) = \frac{W_i U_i^*}{s-\jw_i} + M_\mathrm{reg}(s), \, N(s) = \frac{W_i \Sigma_i V_i^*}{s-\jw_i} + N_\mathrm{reg}(s),\]
    and let \(\Phi_i(s) := U_i^* M_\mathrm{reg}\inv(s) W_i\).  If \(M_\mathrm{reg}(j\omega_i)\) and \(\Phi_i(j\omega_i)\) are invertible, then \(R(\jw_i)\) is finite.
\end{lemma}

\begin{proof}
    We can expand \(N(s)\) in \(R(s) = M\inv(s) N(s)\) and distribute, yielding
    \begin{equation} \label{eq:R_factored}
        R(s) = M\inv(s) \frac{W_i \Sigma_i V_i^*}{s - \jw_i} + M\inv(s)N_\mathrm{reg}(s).
    \end{equation}
    Applying the Woodbury matrix identity to \(M\inv(s)\), we get
    \begin{align*}
        M\inv(s) &= M_\mathrm{reg}\inv(s) - M_\mathrm{reg}\inv(s)W_i((s-\jw_i)I \\
        &+ U_i^*M_\mathrm{reg}\inv(s)W_i)\inv U_i^* M_\mathrm{reg}\inv(s).
    \end{align*}
    Now, right-multiplying by \(W_i\) and simplifying gives back
    \begin{align*}
        M\inv(s)W_i &= M_\mathrm{reg}\inv(s) W_i \brac{I - ((s-\jw_i)I + \Phi_i(s))\inv \Phi_i(s)} \\
        &= M_\mathrm{reg}\inv(s) W_i((s-\jw_i)I + \Phi_i(s))\inv(s-\jw_i),
    \end{align*}
    where \(\Phi_i(s) := U_i^* M_\mathrm{reg}\inv (s) W_i\).  This means the first term of eq.~\eqref{eq:R_factored} can be written as 
    \[M_\mathrm{reg}\inv(s) W_i((s-\jw_i)I + \Phi_i(s))\inv \Sigma_i V_i^*,\]
    noting the cancellation of \(s-\jw_i\).  Under the assumption that \(M_\mathrm{reg}(s)\) and \(\Phi_i(s)\) are invertible at \(s=\jw_i\), we see that 
    \begin{align*}
        \lim_{s\to\jw_i} M\inv(s) &= L_i := M_\mathrm{reg}\inv(\jw_i) \\
        &- M_\mathrm{reg}\inv(\jw_i) W_i \Phi_i\inv(\jw_i) U_i^* M_\mathrm{reg}\inv(\jw_i).
    \end{align*}
    Taking the limit of \(R(s)\) in eq.~\eqref{eq:R_factored} as \(s\to\jw_i\) yields
    \[R(\jw_i) = M_\mathrm{reg}\inv(\jw_i) W_i \Phi_i\inv(\jw_i)\Sigma_i V_i^* + L_i N_\mathrm{reg}(\jw_i),\]
    which is a finite value.
\end{proof}

\begin{remark} \label{rem:Wi_full}
    The hypotheses of Lemma~\ref{lem:reg_cond} also imply that \(W_i\) has full column rank.  Given that \(\Phi_i(j\omega_i) = U_i^* M_\mathrm{reg}\inv(j\omega_i)W_i \in \comp^{r_i\times r_i}\) is invertible, each of its constituent matrices must have rank at least \(r_i\).  Thus \(W_i\in\comp^{p\times r_i}\) must have full column rank, so the hypotheses of Lemma~\ref{lem:reg_cond} imply those of Theorem~\ref{thm:tan_int}.
\end{remark}

\textit{Theorem~\ref{thm:tan_int}}: \thmtanintstatement

\begin{proof}
    Consider the interpolant from equation~\eqref{eq:rtan_int}
    \begin{align*}
            R(s) := &M\inv(s)N(s) := \p{\weights \mathcal{M}(s)}\inv \weights\mathcal{N}(s) \\
            :=&\scalemath{0.85}{\p{I + \sum_{k=1}^\ell \frac{W_k U_k^*}{s-\jw_k}}\inv \p{D + \sum_{k=1}^\ell \frac{W_k \Sigma_k V_k^*}{s-\jw_k}}}.
    \end{align*}
    Let \(G(\jw_i) = U_i\Sigma_i V_i^* + E_i\), where \(U_i^* E_i = 0\).  Then, \(U_i^* G(\jw_i) = \Sigma_i V_i^*\).  \(M\inv\) exists because it is a proper square system with an invertible feedthrough term, thus \(M(s)R(s) = N(s)\) at all but a finite number of points.  Expanding \(M\) and \(N\) and multiplying by \(s-\jw_i\) yields the relation
    \begin{align*}
        &\p{(s-\jw_i)I + W_i U_i^* + \sum_{\mathclap{k\neq i}} \frac{(s-\jw_i)W_k U_k^*}{s-\jw_k}} R(s) =  \\
        &\p{(s-\jw_i)D + W_i \Sigma_i V_i^* + \sum_{\mathclap{k\neq i}} \frac{(s-\jw_i)W_k \Sigma_k V_k^*}{s-\jw_k}}.
    \end{align*}

    Taking the limit of \(s\to\jw_i\) and assuming \(R(j\omega_i)\) is finite, all terms containing \(s-\jw_i\) vanish, giving
    \[W_i U_i^* R(\jw_i) = W_i \Sigma_i V_i^* \implies W_i (U_i^* R(\jw_i) - \Sigma_i V_i^*) = 0.\]
    By assumption, \(W_i\) has full column rank and thus a trivial null-space, thus \(U_i^* R(\jw_i) - \Sigma_i V_i^* = 0\), and \(U_i^* R(\jw_i) = U_i^* G(\jw_i)\).  
    
    \bigskip
    Taking the limit as \(|s|\to\infty\) of equation~\eqref{eq:rtan_int}, all terms in the sum approach zero, thus \(\lim_{|s|\to\infty} R(s) = D\).
\end{proof}

\textit{Theorem~\ref{thm:H_real}}: \thmHrealstatement

\begin{proof}
    Consider the system \(H\) in block column form
    \[H = \begin{bsmallmatrix}H_0 \\ H_1 \\ : \\ H_\ell\end{bsmallmatrix} = \begin{bsmallmatrix}D - G \\ \mathcal{N}_1  - \mathcal{M}_1 G\\ : \\ \mathcal{N}_\ell - \mathcal{M}_\ell G \end{bsmallmatrix}.\]
    The first block system \(H_0\) has the realization
    \[H_0 = \brac{\begin{array}{c|c}A & B \\ \hline -C & 0\end{array}},\]
    Whereas other \(H_k\)s may be written as
    \[H_k = \begin{bsmallmatrix}\mathcal{M}_k & \mathcal{N}_k\end{bsmallmatrix}\begin{bsmallmatrix}-G \\ \phantom{-}I\end{bsmallmatrix}.\]
    The dynamics for this cascade system may be written as
    \begin{align*}
        \dot{x} &= Ax + Bu \\
        \dot{z} &= -U_k^* C x + \jw_k z + (\Sigma_k V_k^* - U_k^* D) u \\
        y &= z.
    \end{align*}
    Now, we utilize the fact that \(U_k^* U_k = I\) and \(U_k \Sigma_k V_k^* = C(\jw_k I - A)\inv B + D  - E_k\) for some \(E_k\) with \(U_k^* E_k = 0\) to simplify the \(z\) state further:
    \begin{align*}
        \dot{z} &= \jw_k I z - U_k^* \p{Cx - (U_k \Sigma_k V_k^* - D)u} \\ 
        &= \jw_k z - U_k^*\p{Cx - C(\jw_k I - A)\inv B u} + \cancel{U_k^* E_k u} \\
        &= \jw_k z - U_k^* C (\jw_k I - A)\inv ((\jw_k I - A)x - Bu) \\
        &= \jw_k z + \underbrace{U_k^* C (\jw_k I - A)\inv}_{C_k} (Ax + Bu - \jw_k x). \\[-2.25em]
    \end{align*}
    Finally, we perform a coordinate change \(w = z - C_k x\) to get 
    \begin{align*}
        \dot{x} &= Ax + Bu \\ 
        \dot{w} &= \jw_k z + C_k (\cancel{Ax + Bu} - \jw_k x) - \cancel{C_k(Ax + Bu)} \\
        &= \jw_k (z - C_k x) = \jw_k w \\
        y &= C_k x + w.
    \end{align*}
    The \(w\) state is uncontrollable, so it can be removed, yielding 
    \begin{equation} \label{eq:Hi1}
        H_k = \brac{\begin{array}{c|c}A & B \\ \hline C_k & 0\end{array}}, \; C_k = U_k^* C(\jw_k I - A)\inv
    \end{equation}
    as a state space realization for \(H_k\).  Because each block system shares the same \(A\) and \(B\) matrices, they can be vertically concatenated together by vertically concatenating their \(C\) and \(D\) matrices, yielding the realization in the Theorem statement.
\end{proof}

\begin{lemma} \label{lem:min_real}
    Let
    \[G = \brac{\begin{array}{c|c}A & B \\ \hline C & D\end{array}} \text{ and } G_{\min} = \brac{\begin{array}{c|c}A_{\min} & B_{\min} \\ \hline C_{\min} & D\end{array}},\]
    where \(G_{\min}\) is a minimal realization of \(G\).  Let \(\Theta\) and \(\Theta_{\min}\) be the matrices \(\frac{1}{2\pi}\int_{-\infty}^{\infty}\theta\theta^*\,\rmd\omega\) of Section~\ref{weights.sec} for \(G\) and \(G_{\min}\) respectively.  Also let
    \[\cC = \begin{bsmallmatrix}\cC_1 \\ : \\ \cC_\ell\end{bsmallmatrix}\; \text{and} \; \cC_{\min} = \begin{bsmallmatrix}\cC_{1,\min} \\ : \\ \cC_{\ell,\min}\end{bsmallmatrix},\]
    where \(\cC_k = U_k^* C(s_k I - A)\inv\) and \(\cC_{k,\min} = U_k^* C_{\min} (s_k I - A_{\min})\inv\).  Then \(\cC \Theta \cC^* = \cC_{\min} \Theta_{\min} \cC_{\min}^*\), \(C\Theta \cC^* = C_{\min}\Theta_{\min} \cC_{\min}^*\) and \(C\Theta C^* = C_{\min}\Theta_{\min}C_{\min}^*\).      
\end{lemma}
\begin{proof}
    First, note the integral forms 
    \begin{align*}
        \cC\Theta \cC^* &= \frac{1}{2\pi}\int_{-\infty}^{\infty} \cC \theta(\jw)\theta^*(\jw)\cC^* \rmd\omega \\
        C\Theta \cC ^* &= \frac{1}{2\pi}\int_{-\infty}^{\infty} C \theta(\jw)\theta^*(\jw)\cC^* \rmd\omega \\
        C\Theta C^* &= \frac{1}{2\pi}\int_{-\infty}^{\infty} C \theta(\jw)\theta^*(\jw)C^* \rmd\omega
    \end{align*}
    where \(\theta(s) = (sI-A)\inv B\).  Consider the similarity transformation \(T\inv\) on \(G\) yielding the canonical Kalman decomposition~\cite{Hespanha09}, i.e. 
    \begin{align*}
        A &= T\inv \begin{bsmallmatrix}A_{\min} & 0 & A_{\times o} & 0 \\ A_{c\times} & A_{c\overline{o}} & A_{\times\times} & A_{\times \overline{o}} \\ 0 & 0 & A_{\overline{c}o} & 0 \\ 0 & 0 & A_{\overline{c}\times} & A_{\overline{c}\,\overline{o}}\end{bsmallmatrix} T, \; B = T\inv \begin{bsmallmatrix}B_{\min} \\ B_{c\overline{o}} \\ 0 \\ 0\end{bsmallmatrix}, \\
        C &= \begin{bsmallmatrix}C_{\min} & 0 & C_{\overline{c}o} & 0\end{bsmallmatrix} T.
    \end{align*}
    Note that the symbol \(\times\) denotes a quantity that is unimportant for the proof.  Applying this transformation on \((sI-A)\inv\) yields 
    \[(sI-A)\inv = T\inv \begin{bsmallmatrix}(sI-A_{\min})\inv  & 0 & \times & 0 \\ \times & \enspace \mathclap{(sI-A_{c\overline{o}})\inv} \enspace\quad & \times & \times \\ 0 & 0 &\enspace\quad\mathclap{(sI-A_{\overline{c}o})\inv} \enspace & 0 \\ 0 & 0 & \times & (sI-A_{\overline{c}\,\overline{o}})\inv\end{bsmallmatrix} T,\]
    thus
    \begin{align*}
        \cC_k &= \begin{bsmallmatrix}U_k^* C_{\min} (s_k I - A_{\min})\inv & 0 & \times & 0\end{bsmallmatrix} T = \begin{bsmallmatrix}\cC_{k,\min} & 0 & \times & 0\end{bsmallmatrix}T, \\
        \cC &= \begin{bsmallmatrix}\cC_{\min} & 0 & \times & 0\end{bsmallmatrix}T, \; \text{and} \\
        \theta(s) &= T\inv \begin{bsmallmatrix}(sI-A_{\min})\inv B_{\min} \\ \times \\ 0 \\ 0\end{bsmallmatrix} =: T\inv \begin{bsmallmatrix}\theta_{\min}(s) \\ \times \\ 0 \\ 0\end{bsmallmatrix}.
    \end{align*}
    From this, it is evident that \(C\theta(s) = C_{\min}\theta_{\min}(s)\) and \(\cC\theta(s) = \cC_{\min}\theta_{\min}(s)\).  Substituting these back into the integrals, factoring out \(C_{\min}\) and \(\cC_{\min}\), and replacing the integrals with \(\Theta_{\min}\) yields the equalities.  
\end{proof}

\begin{lemma} \label{lem:Ch_full}
    Let \(p = 1\), so that \(C\in\comp^{1\times n}\) and each \(U_k\) is a nonzero scalar, and let
    \[\cC = \begin{bsmallmatrix}U_1^* C(s_1I - A)\inv \\ : \\ U_r^* C(s_r I - A)\inv\end{bsmallmatrix}\in\comp^{r\times n},\]
    where \(A\in\comp^{n\times n}\), each \(s_k\) is distinct and not an eigenvalue of \(A\), and the pair \((C, A)\) is observable.  Then \(\rank(\cC) = \min(r, n)\).  Consequently, when \(p = 1\) and \(r \leq n_{\min}\), Assumption~\ref{ass:CTC} holds.
\end{lemma}
\begin{proof}
    Scaling the \(k\)th row by \(U_k^* \neq 0\) does not change the rank, so without loss of generality let every \(U_k = 1\).  Then \(\cC\) satisfies the Sylvester equation \(S\cC - \cC A = \mathbf{1}C\) with \(S := \mathrm{diag}(s_1, \ldots, s_r)\) and \(\mathbf{1}\in\comp^r\) the vector of ones, as is seen row by row from \(s_k C(s_kI - A)\inv - C(s_kI - A)\inv A = C\).  Since \(S\) and \(A\) share no eigenvalues, \(\cC\) is the unique solution, and since the right-hand side has rank one, \cite[Thm.~4]{desouza81} gives \(\rank(\cC) = \min(\mu, \nu)\), where \(\mu\) and \(\nu\) are the ranks of the controllability matrix of \((S, \mathbf{1})\) and the observability matrix of \((C, A)\) respectively.  The former is the Vandermonde matrix of the distinct \(s_k\), so \(\mu = r\), and \(\nu = n\) by observability, thus \(\rank(\cC) = \min(r, n)\).

    Applied to a minimal realization of \(G\) with \(s_k = \jw_k\), this shows \(\cC_{\min}\) has full row rank when \(r \leq n_{\min}\); since \(\Theta_{\min} > 0\) by controllability, Lemma~\ref{lem:min_real} gives \(\cC\Theta\cC^* = \cC_{\min}\Theta_{\min}\cC_{\min}^* > 0\).
\end{proof}

\section*{References}
\bibliographystyle{IEEEtran}
\bibliography{IEEEabrv, paper_bib}

\begin{IEEEbiography}[{\includegraphics[width=1in,height=1.25in,clip,keepaspectratio]{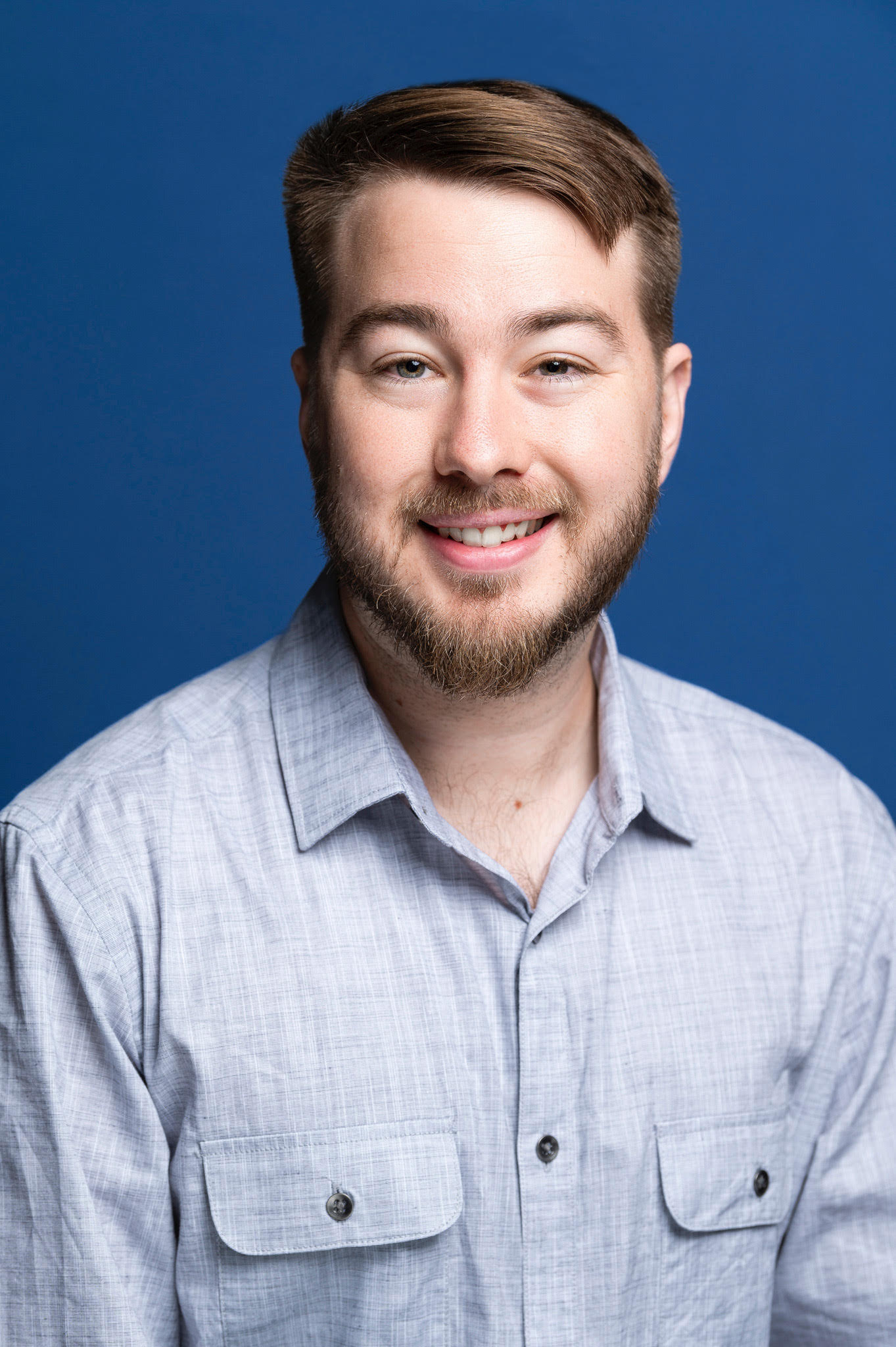}}]{Jared Jonas} (Member, IEEE) received a B.Sc. degree in mechanical engineering from California State University, Fullerton in 2019 and a M.Sc. degree in mechanical engineering from University of California, Los Angeles in 2022.  After, he joined the department of mechanical engineering at University of California, Santa Barbara where he is currently a Ph.D. candidate.  His research interests include computational methods in optimal control, reduced order modeling, and system identification.  
\end{IEEEbiography}

\begin{IEEEbiography}[{\includegraphics[width=1in,height=1.25in,clip,keepaspectratio]{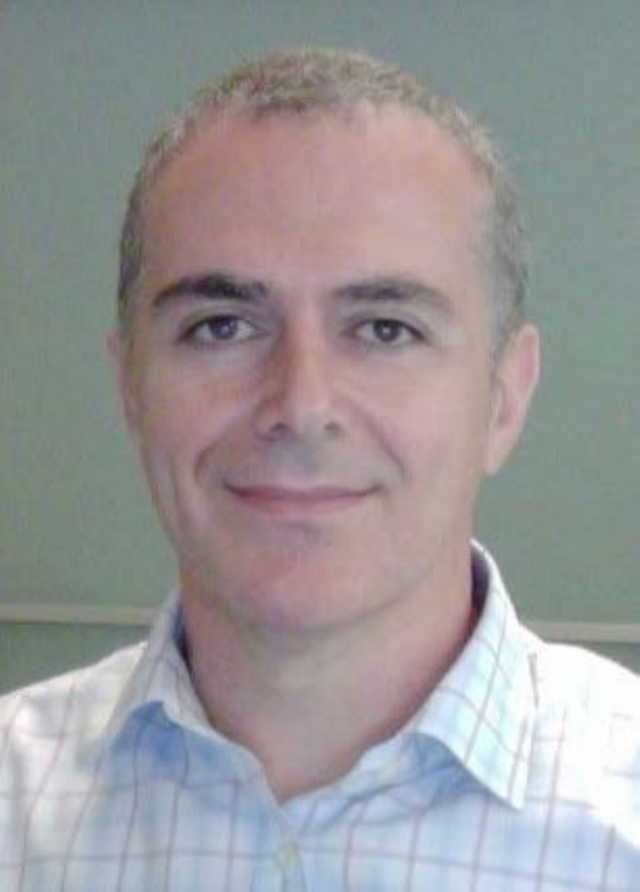}}]{Bassam Bamieh}
(Fellow, IEEE) received the B.Sc. degree in electrical engineering and physics from Valparaiso University, Valparaiso, IN, USA, in 1983, and the M.Sc. and Ph.D. degrees in electrical and computer engineering from Rice University, Houston, TX, USA, in 1986 and 1992, respectively.
From 1991 to 1998, he was an Assistant Professor with the Department of Electrical and Computer Engineering, and the Coordinated Science Laboratory, University of Illinois
at Urbana-Champaign, after which he joined the University of California at Santa Barbara (UCSB), where he is currently a Professor of Mechanical Engineering. His research interests include robust and optimal control, distributed and networked control and dynamical systems, and related problems in fluid and statistical mechanics and thermoacoustics.
Dr. Bamieh is a past recipient of the IEEE Control Systems Society G. S. Axelby Outstanding Paper Award (twice), the AACC Hugo Schuck Best Paper Award, and the National Science Foundation CAREER Award. He was a Distinguished Lecturer of the IEEE Control Systems Society (twice), and is a Fellow of the International Federation of Automatic Control (IFAC).                                                                                     
\end{IEEEbiography}
\end{document}